\documentclass[a4paper]{amsart}
           
\usepackage[english]{babel}
\usepackage{amsfonts,amsthm,amssymb,enumitem,stmaryrd}
\usepackage{amsmath} 
\setlist[itemize] {label=$\pmb\triangleright$}
\usepackage{xfrac}
\usepackage{color}
\usepackage{apalike}
\usepackage{booktabs}

\definecolor{darkgreen}{rgb}{0,0.55,0}

\usepackage{graphicx}
\usepackage{hyperref}
\hypersetup{pdfstartview={Fit},    
    pdftitle={A stochastic partial differential equation model for limit order book dynamics},    
    pdfauthor={Rama Cont and Marvin Mueller},     
    pdfsubject={SPDE model for limit order book dynamics},   
    pagebackref=true,
    colorlinks=true,
    linkcolor=blue,
    citecolor=darkgreen}
\title{A stochastic partial differential equation model for limit order book dynamics}

\author{Rama Cont}

\author{Marvin S. M\"uller}

\date{19th March 2021}                                                                                             
\address[R. Cont]{Mathematical Institute, University of Oxford.}
\address[M. S. M\"uller]{Department of Mathematics, ETH Z\"urich\\\textit{current address:} 2Xideas AG, Seestrasse 39, CH-8700 K\"usnacht}
                     
\thanks{M. S. M. is grateful for generous support from the Swiss National Science Foundation through  SNF grant $205121\_163425$ and from the ETH Foundation at ETH Zurich. The authors thank Martin Keller-Ressel for comments and discussions.}

\numberwithin{equation}{section}

\newcommand\1{\mathbf{1}}

\newcommand\R{\mathbb{R}}

\newcommand\N{\mathbb{N}}

\providecommand{\abs}[1]{\left\lvert#1\right\rvert}
\providecommand{\norm}[2]{\left\lVert#1\right\rVert_{#2}}

\newcommand{\scalp}[3]{\left\langle #1, #2\right\rangle_{#3}}
\newcommand{\iprod}[3]{\left\langle #1, #2\right\rangle_{#3}}
\newcommand{\ip}[3]{\left\langle #1, #2\right\rangle_{#3}}

\newcommand{\PP}{\mathbb{P}}

\newcommand{\EE}{\mathbb{E}}

\newcommand{\sD}{\mathbb{D}}

\renewcommand\d{\,\operatorname{d}\hspace{-0.05cm}}

\renewcommand{\P}[1]{\mathbb{P}\left[#1\right]}                     
\newcommand{\E}[1]{\mathbb{E}\left[#1\right]}                     

\newcommand\Id{\operatorname{Id}}                               
\newcommand\dom{{\rm dom}}                                      

\newcommand\ddt{\tfrac{\partial}{\partial t}}

\newcommand\ddx{\tfrac{\partial}{\partial x}}
\newcommand\ddxx{\tfrac{\partial^2}{\partial x^2}}

\theoremstyle{plain}
\newtheorem{thm}{Theorem}[section]

\newtheorem{prop}[thm]{Proposition}

\newtheorem{cor}[thm]{Corollary}
\newtheorem{lem}[thm]{Lemma}
\theoremstyle{remark}

\newtheorem{rmk}[thm]{Remark}

\newtheorem{model}[thm]{Model}

\theoremstyle{definition}
\newtheorem{df}[thm]{Definition}
\newtheorem{ex}[thm]{Example}
\newtheorem{ass}[thm]{Assumption}

\theoremstyle{example}

\begin{document}
\vspace*{-.1\baselineskip}
\maketitle
\vspace*{-2.5\baselineskip}
\begin{abstract}
We propose an analytically tractable class of models for the dynamics of a limit order book, described through a stochastic partial differential equation (SPDE) with multiplicative noise for the order book centered at the mid-price, along with  stochastic dynamics for the mid-price which is consistent with the order flow dynamics. We provide conditions under which the model admits a finite dimensional realization driven by a (low-dimensional) Markov process, leading to efficient estimation and computation methods.
  We study two examples of parsimonious  models in this class: a two-factor  model and a model with mean-reverting order book depth. For each model we analyze in detail the role of different parameters, the dynamics of the price, order book depth, volume and order imbalance, provide an intuitive financial interpretation of the variables involved and show how the model reproduces statistical properties of price changes, market depth and order flow in limit order markets.
\end{abstract}
$\,$
\vspace*{-.7cm}
\setcounter{tocdepth}{3}

\tableofcontents

\newpage
Financial instruments such as stocks and futures are increasingly traded in electronic, order-driven markets, in which
 orders to buy and sell are centralized in a {\it  limit order book} and
market orders are executed against the best available offers in the limit order book.
The dynamics of prices in such markets are not only interesting from the viewpoint of market participants --for trading and order execution--
but also from a fundamental perspective, since they provide a detailed view of the dynamics of supply fand demand and their role in
 price formation.

The  availability of a large amount of high frequency data on order flow, transactions and price dynamics on these markets has instigated a line of research which, in contrast to traditional market microstructure models which make assumptions on the behavior and preferences of various types of agents, focuses on the {\it statistical modeling} of aggregate order flow and its relation with price dynamics, in a quest to understand the interplay between price dynamics and order flow of various market participants \cite{IEEE}.

A fruitful line of approach to these questions has been to model the stochastic dynamics of the {\it limit order book}, which centralizes all buy and sell orders, either  as a queueing system \cite{luckock,smith2003,cont2010,contlarrard2012,cont2013,kelly2018} or, at a coarse-grained level, through a  (stochastic) partial differential equation  describing the evolution of the distribution of buy and sell orders  \cite{lasrylions,caffarelli2011,burger2013,carmonaWebster,markowich2016,hambly2018,horst2018}. 
These PDE models may be viewed as scaling limits of discrete point process models \cite{contlarrard2012,hambly2018,horst2018}.

Although joint modeling of order flow at all price levels in the limit order book is  more appealling,  (S)PDE models have lacked the analytical  and computational tractability needed for applications; as a result, most analytical results  have  been derived using reduced-form models of the best bid-ask queues \cite{contlarrard2012,cont2013,Chavez2017,rosenbaum2017}. 

We propose  a class of  stochastic models for the dynamics of the limit order book which represent the dynamics of the entire order book while retaining at the same time the analytical and computational
tractability of low-dimensional Markovian models, and provides realistic dynamics for the joint dynamics of the market price and order book depth.  
Starting with a description of the dynamics of the limit order book via a stochastic partial differential equation (SPDE) with multiplicative noise, we show that in many cases, the solutions of this equation may be parameterized in terms of a low-dimensional diffusion process, which then
makes the model computationally tractable. In particular, we are able to  derive analytical relations between
model parameters and various observable quantities. This feature may
 be used for calibrating model parameters to match statistical
features of the order flow and leads to empirically testable 
predictions, which we proceed to test using high frequency time series of order flow in electronic equity markets.

{\bf Outline} 
Section \ref{sec:linearhom} introduces a description of the dynamics of a limit order book through a stochastic partial differential equation (SPDE). We  describe the various terms in the equation, their interpretation and discuss the implications for price dynamics (Section \ref{sec:price}). This class of models is part of a more general family of SPDEs driven by semimartingales, introduced in Sec. \ref{ssec:general_class} and studied in Sec. \ref{sec:generalSPDE}. 

We then focus on two analytically tractable examples: a two-factor  model (Section \ref{sec:2factorLOB}) and a model
with mean-reverting depth and imbalance (Section \ref{sec:meanrev}). For each model we
perform a detailed analysis of the role of different parameters and
study the dynamics of the price, order book depth, volume and order imbalance,  provide
an intuitive financial interpretation of the variables involved and
show how the model may be estimated from  financial time series of
price, volume and order flow.

\begin{figure}[bt]
  \centering
  \includegraphics[width= .9\textwidth]{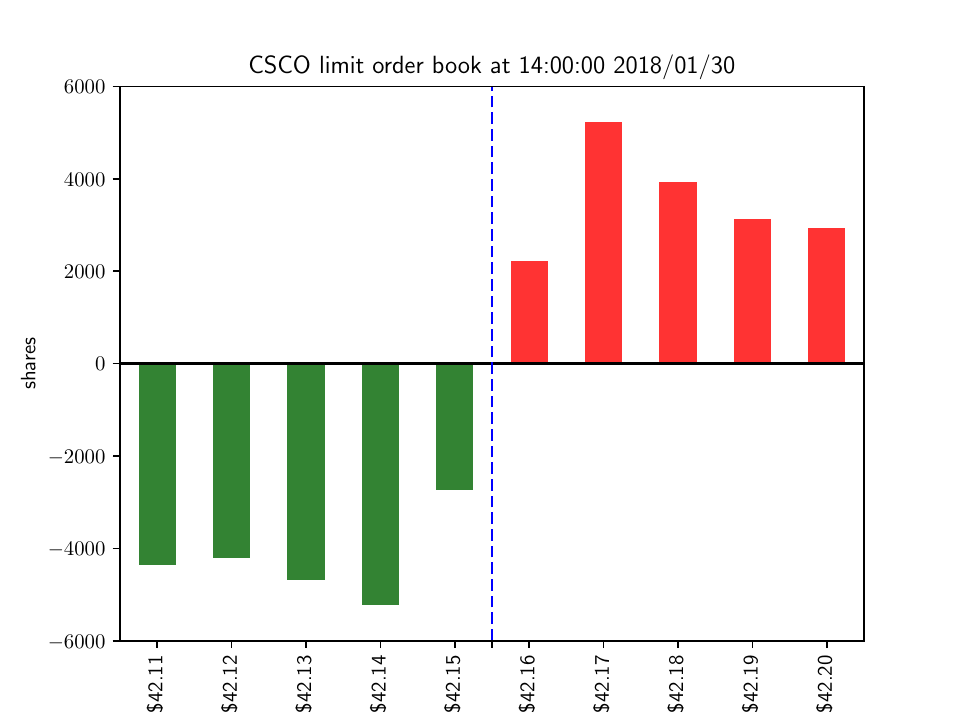}
  \caption{Snapshot of the NASDAQ limit order book for CISCO shares (Jan 30, 2018), displaying
    outstanding buy orders (green) and sell
    orders (red)  awaiting execution at different prices. The highest buying price (\$ 42.15 in this example) is the {\it bid} and the lowest selling price (\$ 42.16)  is the {\it ask}.}
  \label{fig:LOB}
\end{figure}

\section{A stochastic PDE model for limit order book dynamics}
\label{sec:linearhom}

We consider a market for a financial  asset (stock, futures contract, etc.) in which buyers and sellers may submit {\it limit orders} to buy or sell a certain quantity of the asset at a certain price, and {\it market orders} for immediate execution against the best available price.\footnote{In the following we do not distinguish market orders and marketable limit orders i.e. limit orders with a price better than the best price on the opposite side.} Limit orders awaiting execution are collected in the {\it limit order book}, an example of which is shown in Figure \ref{fig:LOB}:  at any  time $t$, the state of the limit order book is summarized by the volume  $V(t,p)$ of orders awaiting execution at  price levels $p$ on a  grid with mesh size given by the minimum price increment or {\it tick } size $\delta$. By convention we associate negative volumes with buy orders and positive volumes with sell orders, as shown in Figure \ref{fig:LOB}. 
An admissible order book configuration is then represented by a function $p\mapsto V(p)$ such that 
$$0<s^b(V):=\sup\{p>0,\quad V(p)<0\} \leq s^a(V):=\inf\{p>0,\quad V(p)>0\}<\infty.$$
$s^b(V)$ (resp. $ s^a(V)$) is called the bid (resp. ask) price and represents the price associated with the best buy (resp. sell) offer. The quantity 
$$S=\frac{s^a(V)+s^b(V)}{2}$$ is called the {\it mid-price} and the difference $s^a(V)-s^b(V)$ is called the bid-ask spread. In the example shown in Figure  \ref{fig:LOB}, $s^b(V)=42.15, s^a(V)=42.16$ and the bid-ask spread is equal in this case to the tick size, which is 1 cent. 

One modelling approach has been to represent the dynamics of $V(t,p)$ as a  spatial (marked) point process \cite{luckock,cont2010,cont2013,kelly2018}. These models preserve the discrete nature of the dynamics at high frequencies but can become computationally  challenging as one tries to incorporate realistic dynamics. In particular, price dynamics,  which is endogenous in such models, is difficult to study, even when the order flow is a Poisson point process.

When the bid-ask spread and tick size $\delta$ are much smaller than the price level, as is often the case,
another modelling approach is to use a continuum approximation for the order book, describing it through its  {\it density} $v(t,p)$ representing the volume of orders per unit price: $$ V(t,p)\simeq v(t,p)\delta.$$ 
The evolution of the density of buy and sell orders is then described through a partial differential equation (PDE).
 A deterministic description of the dynamics of order densities through a system of coupled partial differential equations was proposed by \cite{lasrylions} and studied in detail by \cite{chayes2009,caffarelli2011,burger2013}. In the Lasry-Lions model,   the  evolution of the density of buy and sell orders is described by a pair of diffusion equations coupled through the dynamics of the price, which represents the free boundary between prices of buy and sell orders.  This model is appealing in many respects, especially in terms of analytical tractability, but leads to a deterministic price process which decays to a constant price, so does not provide any insight into the relation between liquidity, depth, order flow and price volatility. \cite{markowich2016} explore some stochastic extensions of this model but essentially show that these extensions do not provide realistic price dynamics.

We adopt here this continuum approach for the description of the limit order book, but describe instead its dynamics through a {\it stochastic} partial differential equation, paying close attention to price dynamics and its relation with order flow.

The model we propose shares some features with \cite{lasrylions}, but also has some essential differences. Unlike the Lasry-Lions model, which is a free boundary problem in which the dynamics of the price  is implicitly determined, we formulate the model as a stochastic partial differential equation in relative price coordinates, which leads to a  {\it stochastic } moving boundary problem in absolute price coordinates.
This leads to a more realistic joint dynamics for the market price and order book depth which can be related to empirical observations. 
Our model also relates to  the classes of models studied in \cite{horst2018,hambly2018} as scaling limits of discrete queueuing systems.
 
We now describe our model in some detail.
 
\subsection{State variables and  scaling transformations}
\label{ssec:scaling}
We focus on the case where the tick size $\delta$ and the bid-ask spread  are small compared to the typical price level and
 consider a  limit order book described in terms of a mid-price $S_t$ and the density $v(t,p)$ of orders at each price level $p$, representing 
 buy orders for $p< S_t$, and 
 sell orders for $p>S_t$.
We use the convention, shown in Figure \ref{fig:LOB},  of representing buy orders with a negative sign and sell orders with a positive sign, so
$$ v(t,p)\leq 0\quad{\rm for}\quad p<S_t\quad{\rm and}\quad v(t,p)\geq 0\quad{\rm for}\quad p>S_t$$

Limit orders are executed  against market orders according to price priority and their position in the queue;  execution of a limit order only occurs if they are located at the best (buy/sell) prices. This means that price dynamics is determined by the interaction of market orders with limit orders of opposite type at or near the {\it interface} defined by the best price \cite{cont2010}. Due to this fact, most limit orders flow are submitted close to the best price levels: the frequency of limit order submissions is highly inhomogeneous  as a function of distance to the best price and concentrated near the best price. As shown in  previous empirical studies, order flow intensity at a given distance from the best price can be considered as a stationary variable in a first approximation \cite{bouchaud2009,cont2010}.
For this reason, in a stochastic description it is more convenient to model the dynamics of order flow in the reference frame of the (mid-)price $S_t$. We define
$$u_t(x)= v(t,S_t+x) $$ where $x$ represents a distance from the mid-price.
We refer to $u_t$ as the {\it centered} order book density.

The simplest way of centering  is to set $x(p)= p - S_t$ but other, nonlinear, scalings may be of interest.
Although limit orders may be placed at any distance from the bid/ask prices, price dynamics is dominated by the behavior of the order book a few levels above and below the mid price  \cite{contlarrard2012}. 
This region becomes infinitesimal if the tick size $\delta$ is naively scaled to zero, suggesting that the correct scaling limit is instead one in which we choose as coordinate a scaled version $(p-S_t)$, as classically done in boundary layer analysis of  PDEs \cite{boundary}, in order to zoom into the relevant region:
\begin{equation}
  \label{eq:scaling}
  x(p):= - (S_t-p)^a,\quad p< S_t,\qquad x(p) = (p-S_t)^a,\quad p>S_t, \qquad a>0
\end{equation}
for bid and ask side, respectively.  We will consider examples of such nonlinear scalings when discussing applications to high-frequency data in Sections 3 and 4.

These arguments also justify limiting the range of the argument  $x$to a bounded interval $[-L,+L]$, setting $u_t(x)=0$ for $x\notin (-L,L)$. This amounts to assuming that no orders are submitted at price levels at distances $|x|\geq L$ from the mid-price and that orders previously submitted at some price $p$ are cancelled as soon as $|S_t-p|\geq L$ i.e. when the mid-price $S_t$ moves away from $p$  by more than $L$. When $L$  is a large multiple of daily volatility, this is a realistic assumption. In some market (for example futures contracts), limit orders can be in fact only submitted within a range $\pm L$ of the mid-price.

\subsection{Dynamics of the centered limit order book}
\label{ssec:setup_dyn}
Empirical studies on intraday order flow in electronic markets reveal the coexistence of two, very different types of order flow operating at different frequencies \cite{lehalle2018}. 

On one hand, we observe the submission (and cancellation) of orders  queueing at various price levels on both sides of the market price by regular market participants. Cancellation may occur in several ways: we distinguish {\it outright } cancellations, which we model as proportional to current queue size, from cancellations with replacement (`order modifications'), in which an order is cancelled and immediately replaced by another one of the same type, usually at a neighboring price limit.
The former results in a net decrease in the volume of the order book whereas the latter is {\it conservative} and simply shifts orders across neighboring levels of the book.
Further decomposing this conservative flow into a symmetric and antisymmetric part leads to two terms in the dynamics of $u_t$: a  {\it diffusion} term representing the cancellation of orders and their (symmetric) replacement by orders  at neighboring price levels and a {\it convection} (or transport) term representing the cancellation of orders and their replacement by orders closer to the mid-price.
Denoting by $\nabla$ the gradient in the variable $x$, the net effect of this order flow on the order book may thus be described as a superposition of  
\begin{itemize}
\item a  term $f^b(x)$ (resp. $f^a(x)\ )$ representing the rate   of buy (resp. sell) order submissions  at a distance $x$ from the best price; 
\item a term $\alpha_b\ u_t(x)$ (resp. $\alpha_a\ u_t(x)$) representing (outright) proportional cancellation of limit buy (resp. sell) orders  at a distance $x$ from the mid-price  (where $\alpha_a, \alpha_b \leq 0$).
\item a convection term  $-\beta_b \nabla u_t(x)$ (resp. $+\beta_a \nabla u_t(x)$) with $\beta_a, \beta_b>0$ which models the replacement of buy (resp. sell) orders  by orders closer to the mid-price (i.e. closer to $x=0$, hence the signs in these terms): in the reference frame where the origin is the mid-price, this translates into a flow of volume towards the origin;
\item a diffusion term $\eta_b \Delta u_t(x)$ (resp. $\eta_a \Delta u_t(x)$) which represents the cancellation and symmetric replacement of orders at a distance $x$ from the mid-price. 
\end{itemize}
Another component of order flow is the one generated by high-frequency traders (HFT). These market participants buy and sell at very high frequency and under tight inventory constraints, submitting and cancelling large volumes of limit orders near the mid-price and  resulting in an order flow  whose net contribution to total order book volume is zero on average over longer time intervals  but whose sign over small time intervals fluctuates at high frequency.  At the coarse-grained  time scale of the average (non-HF) market participants, these features  may be modeled as a multiplicative noise term of the form 
\begin{itemize}
\item 
$\sigma_b u_t(x) dW^b$ for  buy orders ($x<0$) and  $\sigma_a u_t(x) dW^a$ for sell orders ($x>0$) \end{itemize}
where $(W^a,W^b)$ is  a two-dimensional Wiener process (with  possibly correlated components).
The multiplicative nature of the noise accounts for the high-frequency cancellations associated with HFT orders. 

The impact of these different order flow components may be summarized by the following  stochastic
partial differential equation for the centered order book density $u$:
\begin{eqnarray}
      \d u_t(x) &= \left[\eta_a \Delta u_t(x) + \beta_a \nabla u_t(x) + \alpha_a u_t(x) + f^a(x) \right]\d t  + \sigma_a u_t(x)  \d W^a_t ,  \quad\,x \in (0,L), \nonumber\\
    \d u_t(x) &= \left[\eta_b \Delta u_t(x) - \beta_b \nabla u_t(x) + \alpha_b u_t(x) - f^b(x)\right]\d t + \sigma_b u_t(x)  \d W^b_t,  \quad\,x\in (-L,0)\nonumber\\
    & u_t(x)\leq 0,\quad x< 0,\qquad\quad\;
    u_t(x) \geq 0,\quad x>0,  \nonumber \\
    & u_t(0+)  = u_t(0-) = 0, \qquad    u_t(-L) = u_t(L) = 0   \label{eq:lSPDE}
\end{eqnarray}
Here $\eta_a$, $\eta_b$, $\beta_b$, $\beta_a$, $\sigma_a$, $\sigma_b \in (0,\infty)$, $\alpha_a, \alpha_b \leq 0$ and $f^a$, $f^b
\colon I \to [0,\infty)$ although the equation may be equally considered without these sign restrictions.

Note that, unlike the Lasry-Lions model, there is no `smooth pasting' condition at $x=0$: in general $\nabla u_t(0+)\neq \nabla u_t(0-)$: the difference $\nabla u_t(0+)- \nabla u_t(0-)$ is in fact random and represents an imbalance in the flow of buy and sell orders, which drives price dynamics. This important feature is discussed in Section \ref{sec:price} below.

\begin{rmk}\label{remark.zero} In simple price impact models used in the literature on optimal trade execution it is assumed that the relation between price impact and order size is deterministic. This corresponds to the case $\alpha u+ f=\beta=\sigma=\eta=0$ which leads  to a constant centered order book profile $u_t(.)=u_0(.)$.
These terms thus correspond to  deformations of the centered order book profile due to new order book events and lead to a stochastic market impact of trades dependent on the current state of the order book.
\end{rmk}

The existence of a solution  satisfying the boundary and sign constraints is not obvious but we will see in Section \ref{sec:generalSPDE} that  \eqref{eq:lSPDE} is well-posed:
it follows from \cite[Theorem 6.7]{dPZinf} and \cite[Theorem 3]{milianComp} 
that, when $f_a$,
$f_b\in L^2(I)$, then for all $u_0\in L^2(I)$ there exists a unique  weak
solution of~\eqref{eq:lSPDE} (see Definition~\ref{df:weak_solution}
below) and, when $u_0 \vert_{(0,L)}\geq 0$ and  $u_0\vert_{(-L,0)}\leq 0$ this solution satisfies
\begin{equation}
  \label{eq:45}
  u_t \vert_{(0,L)} \leq 0,\qquad u_t\vert_{(-L,0)}\geq 0.
\end{equation}
We will study the mathematical properties of the solution in more detail below. 

\subsection{Price dynamics}\label{sec:price}

The dynamics of the limit order book determines the dynamics of the bid and ask price, which corresponds to the location of the best (buy and sell) orders. The dynamics of the price should thus be related to the arrival and execution of orders in the order book.

To understand the relation between price dynamics and order flow, let us take a step back and consider   an order book with discrete price levels, multiples of a tick size $\delta$, $D^b$ orders per level on the bid side and  $D^a$ orders per level on the ask side. 
Price changes during a time interval $[t, t+\Delta t]$ are triggered through the interaction of the net order flow, or {\it order flow imbalance} (OFI) and the outstanding limit orders at the top of the order book \cite{contImbalance}. 
As illustrated in  Figure~\ref{fig:impact}, an order flow imbalance of $\Delta D^a_t>0$  on the ask side  over a
short time interval $[t,t+\Delta t]$ represents an excess of buy orders, which  will then be executed against limit sell orders sitting on the ask side and move the ask price by
$\Delta D^a_t/D^a$ ticks, resulting  in a price move of $\delta \ \Delta D^a_t/D^a$. Similarly, an order flow imbalance $\Delta D^b_t$ on the bid side will move the bid price up by $\Delta D^b_t/D^b$ ticks.
Using our sign conventions for buy/sell volumes, this leads to the following dynamics: 
$$
     \Delta s^b_t = \delta \frac{\Delta D_t^b}{D^b_t}\qquad  \Delta s^a_t  =  -\delta \frac{\Delta D_t^a}{D^a_t},
$$
so the dynamics of  the mid price $s_t=(s^b_t +s^a_t )/2$ is given by
\begin{equation}
  \label{eq:discreteprice}
  \Delta S_t =  \frac{\delta}{2}\left(\frac{\Delta D_t^b}{ D^b_t} - \frac{\Delta D_t^a}{ D^a_t}\right).
\end{equation}
This relation is exact (up to rounding) in the case of a discrete order book with constant depth per level (and thus, no empty levels), as shown in Figure \ref{fig:impact}. However, in a dynamic setting where the order book may have an arbitrary profile which randomly shifts at each instant, one can only expect a `homogenized' version of \eqref{eq:discreteprice} to hold:
\begin{equation}
  \label{eq:price}
  \Delta S_t =  {\theta}\left(\frac{\Delta D_t^b}{ D^b_t} - \frac{\Delta D_t^a}{ D^a_t}\right)
\end{equation}
where $\theta$ is an {\it impact} coefficient which relates order imbalance to price movements.
This relation between order flow imbalance and price movements has been empirically verified in equity markets \cite{contImbalance}, and we shall use it as a basis for defining the relation between price dynamics and order flow in our model.
\begin{figure}[h!]
\noindent
\centering
\includegraphics[width=0.8\textwidth]{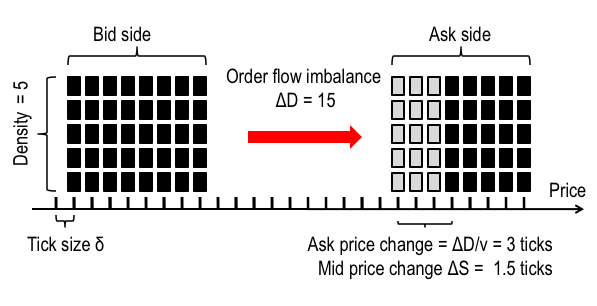}
\caption{Impact of order flow  imbalance on the order book and the price.}
\label{fig:impact}
\end{figure}

Let us now see how the relation \eqref{eq:price} translates in terms of the variables in our model.
Denoting by $D^b_t$  (resp. $D^a_t$) the volume of buy (resp. sell)
limit orders at the top of the book (i.e. the first or average of the first few levels). 
Given a
mid-price $S\in \R_+$, we define a scaling transformation $x\colon [S,S+L]\to [0,\infty)$  as discussed in Section~\ref{ssec:scaling}, with continuously differentiable
inverse and such that $x(S) = 0$. The volume  $D^a$   in the best ask  queue is then given by
\begin{equation}
  \label{eq:depthAskInt}
   D^a = \int_{s}^{s+\delta} u(x(p)) \d p = \int_0^{x(s+\delta)}u(y)
  (x^{-1})'(y)\d y .
\end{equation}
$D^b$ may be similarly defined for the bid side.
These quantities represent the depth at the top of the book; we will refer to them as `market depth'.
 In the case of linear scaling $x(p) = p-S$, using $u(0)=0$ a second order expansion in $\delta>0$ yields
\begin{equation}
  D^a= \int_0^\delta u(x)\d x \approx \delta u(0+) + \frac{\delta^2}{2} \nabla
  u(0+) = \frac{\delta^2}{2} \nabla u(0+). \label{eq:depthAsk}
\end{equation}
Similarly, for the bid side
\begin{equation}
  D^b \approx \frac{\delta^2}{2} \nabla u\left(0-\right).  \label{eq:depthBid}
\end{equation}
Substituting these expressions in \eqref{eq:price}, we obtain the following 
 dynamics of the mid-price:
\begin{equation}
  \label{eq:SDEprice}
  \d S_t =\theta \left(\frac{\d D^b_t}{D_t^b} - \frac{\d
      D^a_t}{D_t^a}\right)=\theta \left(\frac{\d \nabla u_t(0-)}{\nabla u_t(0-)}  - \frac{\d\nabla u_t(0+)}{\nabla u_t(0+)}\right).
\end{equation}
We observe that price dynamics is entirely determined by the order flow at the top of the book and the depth of the limit order book around the mid-price.
The tick size $\delta$, used in the derivation, does not appear anymore in \eqref{eq:SDEprice}. The only trace of the microstructure is  the impact coefficient $\theta$ which relates the order flow imbalance to the magnitude of the price change, and whose amplitude may vary across assets.
\begin{rmk}
Equation~\eqref{eq:SDEprice} requires left and right-differentiability of $u$ at the
origin. This can be
guaranteed whenever $u_t$ takes values in the Sobolev space
$H^{2\gamma}(I)$, for some $\gamma > \sfrac34$ which will be the case in
our  model. Note however that, in contrast to \cite{lasrylions}, in general $\nabla u(0+)
\neq \nabla u(0-)$: the difference between these two quantities is proportional to the order flow imbalance which drives price moves. 
\end{rmk}

\begin{rmk} As noted in Remark \ref{remark.zero}, in the case $\alpha u+f = \beta=\sigma=\eta=0$ corresponds to a constant centered order book profile $u_t=u_0$. In this case, 
Equation~\eqref{eq:SDEprice} implies $dS_t=0$ i.e. the price is constant, which is consistent with a zero net order flow. This is a (desirable) consequence of the consistency between the price dynamics  \eqref{eq:SDEprice} and the order book dynamics \eqref{eq:lSPDE}. 
\end{rmk}
\subsection{Dynamics in  absolute price coordinates} 
\label{ssec:2factor_noncentered}
The model above describes dynamics of the order book in {\it relative} price coordinates, i.e. as a function of the (scaled) distance from the mid-price. 
The density of the limit order book parameterized by the (absolute) price level 
$p\in \R$  is given by
\begin{equation}
  \label{eq:1}
  v_t(p) = u_t(p - S_t),\qquad x\in \R,
\end{equation}
where we extend $u_t$ to $\R$ by setting $u_t(y) = 0$ for
$y\in \R\setminus [-L,L]$.
Assume $S_t$ follows an (arbitrary) It{\^o process
\[ \d S_t =\theta \mu_t \d t +  \theta  \xi_t^b \d W_t^b - \theta \xi_t^a \d W_t^a\]
where $\theta >0$ and $\mu_t$ is predictable and integrable and $\xi_t^a$ and $\xi_t^b$ are predictable and square-integrable processes. This  includes the case of price dynamics \eqref{eq:SDEprice}, which can be used to  express $\mu_t, \xi_t^a, \xi_t^b $ in terms of $u_t$ and model parameters. We will not go into such detail here but will return to this in the examples in Section \ref{sec:2factorLOB} and \ref{sec:meanrev}.
Define
\[ \hat \xi_t := \sqrt{(\xi_t^b)^2 + (\xi_t^a)^2 - 2\varrho_{a,b} \xi_t^b \xi_t^a},\qquad t\geq 0.\]

Using a (generalized) It\^{o}-Wentzell formula (see Appendix~\ref{A:trafo}), we can show that $v$ is the solution of a {\it stochastic moving boundary problem}  \cite{mueller2016stochastic}:
  \begin{multline}
    \label{eq:vts+}
    \d v_t(p) = \left[(\eta_a+\tfrac12 \theta^2 \hat \xi_t^2) \Delta v_t(p)\right.\\
    \left.+ \left(   \beta_a - \theta \mu_t -\theta \sigma_a (\varrho_{a,b} \xi_t^b  - \xi_t^a) \right)\nabla v_t(p)  + \alpha_a v_t(p)\right] \d t \\
    + \left( \sigma_a v_t(p)  +  \theta  \xi_t^a \nabla v_t(p)\right) \d W_t^a  -\theta\xi_t^b \nabla v_t(p) \d W_t^b,    
  \end{multline}
  for $p \in (S_t, S_t+L)$, and
  \begin{multline}
    \label{eq:vts-}
    \d v_t(p) = \left[(\eta_b+\tfrac12 \theta^2 \hat \xi_t^2) \Delta v_t(p)\right.\\
    \left.+ (- \theta \mu_t  - \beta_b  -  \theta\sigma_b (\xi_t^b-\varrho_{a,b}  \xi_t^a)) \nabla v_t(x)  + \alpha_b v_t(p)\right] \d t \\
   + \theta   \xi_t^a \nabla v_t(p)\d W_t^a     + \left(\sigma_b  v_t(p) -\theta \xi_t^b \nabla v_t(p)  \right) \d W_t^b  
 \end{multline}
  for $x\in (S_t-L, S_t)$ with the moving boundary conditions
  \begin{equation}
    \label{eq:bdryst}
    v_t(S_t) = 0,\qquad v_t(y) = 0,\qquad \forall y\in \R\setminus(S_t-L,S_t+L),
  \end{equation}
We refer to \eqref{eq:bdryst} as a {\it stochastic boundary condition} at $S_t$.

Here, we assumed for simplicity that $f^a$, $f^b = 0$. A more detailed discussion of this result is given in Appendix~\ref{A:trafo}.

\subsection{Linear evolution models for order book dynamics}
\label{ssec:general_class}
We will now describe a more general class of linear models for order book dynamics, rich enough to cover the examples we discussed so far, but also covering all level-1 models where the best bid and ask queue are modeled by positive semimartingales. Generally, the densities of orders in the bid and ask side will take values in
some function spaces $H^b$ and $H^a$, respectively. We assume that
orders at relative price level $x$ for $\abs{x} \geq L\in
(0,\infty]$ will be cancelled. The relative price levels are on the bid side $I^b:= (-L,0)$, and on the ask side $I^a:= (0,L)$. Then, in order to preserve the
interpretation of a density it will be reasonable to ask $H^b\subset
L^1_{loc}(I^b)$ and $H^a\subset L^1_{loc}(I^a)$. From mathematical side, we will assume that $H^a$ and $H^b$ are real separable Hilbert spaces. For notational
convenience we now also set $I:= I^b\cup I^a$. 

The density of limit orders at relative price level $x$ and time $t$
is given by $u\colon I \times [0,\infty) \times \Omega \to \R$, such
that $u^\star := u\vert_{I^\star}$ is an $H^\star$-valued adapted process. The
initial state is described by $h\colon I\to \R$, such that $h^\star :=
h\vert_{I^\star}$ is an element in $H^\star$. The (averaged) intra-book dynamics are modeled by linear operators $A_\star\colon \dom(A_\star) \subset H^\star \to H^\star$, for $\star\in
\{a,b\}$, which we assume to be densely defined and such that for $\star \in\{ a,b\}$   there exist weak solutions in $H^\star$ of the equations
\begin{equation}
  \label{eq:gModelHomDet}
  \ddt g_t^\star(h^\star) = A_\star g_t^\star(h^\star),\quad t> 0,\quad g_0^\star(h^\star) = h^\star,
\end{equation}
for each initial state $h^\star\in H^\star$.

The random order arrivals and cancellations are assumed to be proportional and are modelled by cadlag semimartingales $X^b$ and $X^a$, which we assume to have jumps greater than $-1$ almost surely. We assume the initial
order book state is denoted by $h\in H$ and we write $h^a :=
h\vert_{I^a}$, $h^b := h\vert_{I^b}$. 

\begin{model}[Linear Homogeneous Evolution]\label{model:ghom}
  The  general form of the linear homogeneous model is
  \begin{equation}
    \label{eq:gModelHom}
    \left\{\begin{alignedat}{2}
      \d u_t^b &= A_b u_{t-}^b  \d t + u_{t-}^b \d X_t^b,\quad& \text{on }&I^b,\\
      \d u_t^a &= A_a u_{t-}^a  \d t + u_{t-}^a \d X_t^a,\quad& \text{on }&I^a,     
    \end{alignedat}\right.
  \end{equation}
  for $t\geq 0$, and $u_0 = h$.  $u$ can be alternatively expressed as
  \begin{equation}
    \label{eq:10}
    u_t = g_t^b (h^\star)\mathcal{E}_t(X^b) \1_{I^b} + g_t^a(h^\star) \mathcal{E}_t(X^a) \1_{I^a},
  \end{equation}
  where $g^b$ and $g^a$ are solutions of~\eqref{eq:gModelHomDet}, see Theorem~\ref{thm:homo} below. If, in addition, $t\mapsto \nabla g_t^b(0-)$ and $t\mapsto \nabla
  g_t^a(0+)$ are of bounded variation, then we obtain the price dynamics \eqref{eq:SDEprice}.
\end{model}

\begin{cor}
  Assume the setting of Model \ref{model:ghom} and, in addition, that $h^\star$
  is an eigenfunction of $-A_\star$ with eigenvalue $\nu_\star\in \R$, for $\star=b$ and
  $\star=a$. Then,~\eqref{eq:gModelHomDet} can be solved
  explicitly and 
  \begin{equation}
    \label{eq:13}
    u_t = h^b e^{-\nu_b t} \mathcal{E}_t(X^b) \1_{I^b} + h^a e^{-\nu_a
      t} \mathcal{E}_t(X^a) \1_{I^a}. 
  \end{equation}
\end{cor}
\begin{rmk}
  In case that $X^b$ and $X^a$ are (local) martingales, the eigenvalues $-\nu_b$ and $-\nu_a$ play the role of  net order arrival rates on bid and ask side, respectively. 
\end{rmk}

\begin{model}[Linear models with source terms]\label{model:ginhom}
  A more realistic setting assumes in addition an influx/outflow
  of orders at a rate $f^a(x),f^b(x)$ which
  depends on the distance $x$ to the mid price \cite{cont2010}. The equation then becomes:
  \begin{equation}
    \label{eq:gModelInHom}
    \left\{\begin{alignedat}{2}
    \d u_t^b &= \left(A_b u_{t}^b+ f^b\right) \d t +  u_{t}^b \d X_t^b,\quad& \text{on }&I^b,\\
    \d u_t^a &= \left(A_a u_{t}^a+ f^a\right)  \d t +  u_{t}^a \d X_t^a,\quad& \text{on }&I^a,            
    \end{alignedat}\right.
  \end{equation}
  for $t\geq 0$, with initial condition $u_0 = h$.

  As we will discuss in Section \ref{sec:meanrev}, an interesting case is when
 $f^b$ (resp. $f^a$) is an eigenfunction of $-A^b$ (resp.
  $-A^a$) associated with some eigenvalue $\nu_b$ (resp. $\nu_a$). Then by Theorem~\ref{thm:inhomo} we obtain
  \begin{equation}
    \label{eq:10}
    u_t = \left(g_t^b(h^b-f^a)\mathcal{E}_t(X^b) + f^b Z_t^b\right) \1_{I^b} + \left(g_t^a(h^a-f^a) + f^a Z_t^a\right) \1_{I^a},
  \end{equation}
where, for $\star \in \{a,b\}$, $Z_t^\star$ is  the solution of 
  \begin{equation}
    \label{eq:16}
    \d Z_t^\star =  (1-\nu_\star Z_{t-}^\star)\d t + Z_{t-}^\star \d X_t^\star,\quad t\geq 0,\quad Z_0^\star = 1.
  \end{equation}
\end{model}
\begin{rmk}
  If $\nu_b$, $\nu_a>0$ the state of the order book is mean
  reverting to the state $f^b 1_{[-L,0)} + f^a\ 1_{(0,L]}$. We will give an example of such a mean-reverting order book  model in Section~\ref{sec:meanrev}.
\end{rmk}

\begin{rmk}
Any model for the dynamics of the order book implies a model for price dynamics via~\eqref{eq:SDEprice}. In particular this implies a relation between price volatility and parameters describing order flow, in the spirit of \cite{cont2013}. We will derive this relation for the examples studied in the sequel and use it to construct a model-based intraday volatility estimator.
\end{rmk}

In the next section, we will study this class of models from a mathematical point of view. We will then continue with the analysis of the two  examples mentioned above in Sections~\ref{sec:2factorLOB} and~\ref{sec:meanrev}. 

\section{Linear stochastic PDE models with multiplicative noise}\label{sec:generalSPDE}

In order to  further  study the properties of the  SPDE model \eqref{eq:lSPDE}, we require a more explicit characterization of the solution, in order to compute various quantities of interest and estimate model coefficients from observations. A useful approach is to look for a { finite dimensional realization} of the infinite-dimensional process $u$:
\begin{df}[Finite dimensional realizations]\label{df:fdr}  A process  $u= (u_t)_{t\geq 0}$ taking
  values in an (infinite-dimensional) function space $E$ is said to
  admit a {\it finite dimensional realization} of dimension $d\in
  \N$ if there exists an $\R^{d}$-valued stochastic  process $Z=(Z^1,...,Z^d)$
  and a map 
  \[ \phi: \R^{d} \to E\quad{\rm such\ that}\quad \forall t\geq 0, \quad u_t = \phi(Z_t).\]
\end{df}
Availability of a finite dimensional realization for the SPDE~\eqref{eq:lSPDE} makes simulation, computation and estimation problems more tractable, especially if the
process $Z$ is a low-dimensional Markov process.
Existence of such finite-dimensional realizations for stochastic PDEs  have been investigated for SPDEs arising in filtering \cite{levine1991} and interest rate modelling \cite{filipovic2003,gaspar2006}.

We will now show that finite dimensional realizations may indeed be
constructed for a class of SPDEs which includes \eqref{eq:lSPDE}, and
use this representation to perform an analytical study of these models.

\subsection{Homogeneous equations}
We now consider a more general class of linear homogeneous evolution equations
with multiplicative noise taking values in a real separable Hilbert space $(H, \langle \cdot,\cdot\rangle_H)$. Typically, $H$ will be a function space such as $L^2(I)$ for some interval $I\subset \R$. We consider the following class of evolution equations:
\begin{gather}
  \label{eq:gSPDE}
  \begin{split}
    \d u_t &= A u_{t-} \d t + u_{t-} \d X_t,\qquad t>0,\\
    u_0 &= h_0 \in H.
  \end{split}
\end{gather}
where $X$ is a real c\`{a}dl\`{a}g semimartingale whose jumps satisfy $\Delta X_t>-1$ a.\,s. and
$A\colon \dom(A) \subset H \to H$ a linear operator on $H$ whose adjoint we denote by $A^*$. We assume
that $\dom(A)\subset H$ is dense, and $A$ is closed. Since $A$ is closed we have that
also $\dom(A^*)\subset H$ is dense and that $A^{**} = A$ 
\cite[Theorem VII.2.3]{yosida1995functional}.
\begin{df}\label{df:weak_solution}
  An adapted $H$-valued stochastic process
$(u_t)$ is an \emph{(analytical) weak solution} of~\eqref{eq:gSPDE} with
initial condition $h_0$ if, for all $\varphi \in \dom(A^*)$,  $[0,\infty) \ni t\mapsto \langle u_t, \varphi\rangle_H \in \R$ is c\`{a}dl\`{a}g a.\,s. and for each $t\geq 0$, a.\,s.
\[\langle u_t,\varphi \rangle_{H} - \langle h_0, \varphi \rangle_{H} = \int_0^t \langle u_{s-}, A^*\varphi \rangle_{H} \d s+ \int_0^t \langle u_{s-}, \varphi \rangle_{H} \d X_s.\]
\end{df}
The case $X\equiv 0$ corresponds to a notion of weak
solution for the PDE:
\begin{equation}
  \label{eq:gPDE}
  \forall t>0,   \ddt g_t = A g_t\qquad 
    g_0 = h_0.
\end{equation}
That is, for all $\varphi \in \dom(A^*)$, 
\begin{equation}
  \label{eq:12}
  \langle g_t, \varphi\rangle_{H} - \langle h_0,\varphi\rangle_{H}= \int_0^t \langle g_s, A^* \varphi \rangle_{H} \d s,
\end{equation}
where the integral on the right hand side is assumed to exist.\footnote{Note that this slightly differs from the classical
  formulation of weak solutions for  PDEs.} In
particular, this yields that $[0,\infty) \ni t\mapsto \langle g_t,
\varphi\rangle_H\in \R$ is continuous.
\begin{rmk}
  By considering bid and ask side separately, we can
  bring~\eqref{eq:lSPDE} into the form of \eqref{eq:gSPDE}, where $X$
  is a Brownian motion and $A$ is given by $A := \eta\Delta \pm
  \beta\nabla + \alpha \operatorname{\Id}$ on $H:= L^2(I)$, $I :=
  (0,L)$ or $I:= (-L,0)$, with domain
  \[\operatorname{dom}(A):= H^2(I)\cap H^1_0(I),\]
  where  $H^1_0(I)$ is the closure in $H^1(I)$ of test functions with compact support in $I$. 
\end{rmk}
Denote by $Z_t=\mathcal{E}_t(X)$ the stochastic exponential of $X$. We recall the following useful lemma (see e.g. \cite[Lemma 3.4]{karatzas2007numeraire}):
\begin{lem}\label{lem:ExpRec}
  Let 
  \begin{equation}
    \label{eq:10}
    Y_t := -X_t + \left[ X, X\right]_t^c + \sum_{s\leq t} \frac{(\Delta X_s)^2}{1 + \Delta X_s},\qquad t\geq 0,
  \end{equation}
  Then, $\mathcal{E}_t(X) \mathcal E_t(Y) = 1$ almost surely, for all $t\geq 0$. Moreover,
  \begin{equation}
    \label{eq:11}
    \left[ X, Y\right]  = -\left[X,X\right]^c- \sum_{s\leq \cdot} \frac{(\Delta X_s)^2}{1+\Delta X_s}.
  \end{equation}
\end{lem}
\begin{thm}\label{thm:homo}
  Let $Z := \mathcal{E}(X)$,  $h_0 \in H$. Then every
  weak solution of~\eqref{eq:gSPDE} is of the form 
  $$u_t:= Z_tg_t,\qquad t\geq 0$$ where $g$ is a weak solution of~\eqref{eq:gPDE}. 
\end{thm}
\begin{rmk}
  In particular, the SPDE~\eqref{eq:gSPDE} admits a two dimensional realization in the sense of Definition~\ref{df:fdr} with factor process $(t,\mathcal{E}_t(X))$ and $\phi(t,y):= y g_t$.
\end{rmk}
\begin{proof}
  Set $u_t := g_t Z_t$, $t\geq 0$, and for $\varphi \in
  D(A^*)$ write $B^\varphi_t:= \langle g_t,\varphi \rangle_{H}$,
  $C^\varphi_t := B^\varphi_t Z_t = \langle u_t,\varphi \rangle_{H}$. Since $t\mapsto \langle g_t,\varphi \rangle_H$ is continuous and $Z$ is scalar and c\`{a}dl\`{a}g, we get that $t\mapsto \langle u_t, \varphi \rangle_H$ is c\`{a}dl\`{a}g. Note that $B^\varphi$ is of
  finite variation and $Z$ is a semimartingale, so that also $C^\varphi$ is a
  semimartingale. Moreover, by It\^{o} product rule and since $B^\varphi$ is of
  finite variation and continuous,
  \begin{equation}
    \label{eq:10}
    \d C^\varphi_t = B^\varphi_{t} \d Z_t + Z_{t-} \d B^\varphi_t 
    = B^\varphi_{t-} Z_{t-} \d X_t +  \langle u_{t-}, A^* \varphi
    \rangle_{H} \d t,
  \end{equation}
  which is~\eqref{eq:gSPDE}. 
    Now, let $u$ be a solution of \eqref{eq:gSPDE} and set 
  \begin{equation*}
    Y := -X + \left[ X, X\right]^c + J,\qquad J :=\sum_{s\leq \cdot} \frac{(\Delta
      X_s)^2}{1+ \Delta X_s},
  \end{equation*}
  and $Z_t := \mathcal{E}_t(Y)$, $t\geq 0$. Recall that by Lemma~\ref{lem:ExpRec}
  we have $Z_t\mathcal{E}_t(X) = 1$ for all $t\geq 0$. 
  Set $g_t := Z_t u_t$, and, as above, fix $\varphi \in \dom(A^*)$ and
  write $B^\varphi_t:= \langle u_tZ_t, \varphi\rangle_{H}  = \langle g_t,\varphi
  \rangle_{H}$ and $C^\varphi_t := \langle u_t,\varphi \rangle_{H}$. By It\^{o}'s
  product rule and Lemma~\ref{lem:ExpRec},
  \begin{align}\nonumber
    \d B^\varphi_t &= C^\varphi_{t-} \d Z_t + Z_21{t-} \d C^\varphi_t + \d \left[ C^\varphi, Z \right]_t\\
    \nonumber &= C^\varphi_{t-}Z_{t-} \d Y_t + Z_{t-} \langle u_{t-},
                A^*\varphi \rangle_{H} \d t + C^\varphi_{t-}Z_{t-} \d X_t + C^\varphi_{t-}
                Z_{t-} \d \left[ X,Y\right]_t\\
    \nonumber &= \langle g_{t-},
                A^*\varphi \rangle_{H} \d t 
    + B^\varphi_{t-} \left(\d \left[ X,X\right]^c_t + \d J_t\right)
                - B^\varphi_{t-} \left( \d \left[X,X\right]^c_t + \d
                J_t\right)\\
    \nonumber &= \langle g_{t-},
                A^*\varphi \rangle_{H} \d t.
  \end{align}
  Thus, $g$ is a weak solution of~\eqref{eq:gPDE}.
\end{proof}

\begin{ex}
  Let $A$ be the generator of a strongly continuous semigroup
  $(S_t)_{t\geq 0}$. Then, for $h_0\in H$ define
  \[ g_t := S_t h_0,\quad t\geq 0, \]
  which is a weak solution of~\eqref{eq:gPDE}. By Theorem~\ref{thm:homo}.
  \[u_t := \mathcal{E}_t(X) S_t h_0, \qquad t\geq 0, \]
  is a weak solution of~\eqref{eq:gSPDE}.
\end{ex}

\begin{rmk}
  If $h_0$ is an eigenfunction of $A$ with eigenvalue $\nu$, then,
  $g_t = e^{\nu t} h_0$ is the unique locally $H$-integrable solution
  of~\eqref{eq:gPDE}, and the unique solution of \eqref{eq:gSPDE} is given by 
  \[ u_t := h_0 e^{\nu t} \mathcal{E}_t (X).\]
\end{rmk}

\subsection{Inhomogeneous equations}
We keep the assumptions on $A$, $h_0$ and $X$ from the previous
section and let $f\in H$. We now consider the inhomogeneous
linear evolution equations
\begin{gather}
  \label{eq:gSPDEinh}
  \begin{split}
    \d u_t &= \left[ A u_{t} + \alpha f \right] \d t +   u_{t-} \d X_t,\qquad t\geq 0,\\
    u_0 &= h_0.
  \end{split}
\end{gather}
\begin{df}
  A weak solution of~\eqref{eq:gSPDEinh} is an adapted
  $H$-valued stochastic process $u$ such that for all $\varphi \in
  \dom(A^*)$ the mapping $[0,\infty)\ni t\mapsto \iprod{u_t}{\varphi}{H}$
  is c\`{a}dl\`{a}g and 
  \[\langle u_t,\varphi \rangle_{H} - \langle h_0, \varphi \rangle_{H} =
    \int_0^t \langle u_{s-}, A^*\varphi \rangle_{H} \d s+ \int_0^t
    \langle u_{s-}, \varphi \rangle_{H} \d X_s + t\alpha \langle
    f,\varphi\rangle_H,\quad t\geq 0,\]
  almost surely. 
\end{df}
We exclude the cases $\alpha = 0$ or $f\equiv 0$ which correspond to
the homogeneous case discussed above. Let us first consider
the case where $A$ admits at least one eigenfunction. 
\begin{thm}\label{thm:inhomo}
  Suppose that $f\in \dom(A)$ is an eigenfunction for $A$ with eigenvalue $\lambda \in \R$, and let $z_0 >0$ and $Z$ be the solution of
  \begin{equation} 
    \label{eq:Z}
    \d Z_t = \left(\lambda Z_{t-} 
      +  \alpha\right) \d t + Z_{t-} \d X_t,\quad t\geq 0, \quad Z_0 = z_0. 
  \end{equation}
  Then:
  \begin{enumerate}[label=(\roman*)]
  \item\label{inhomo:i:1factor} 
    The stochastic process defined by $u_t = Z_t f$, $t\geq 0$, is a
    solution of \eqref{eq:gSPDEinh} with initial condition $h_0 := z_0
    f$.
  \item\label{inhomo:i:2factor} Let, in addition, $h_0\in H$ be such that there exists a
    weak solution $g = (g_t)_{t\geq 0}$ of the deterministic equation
    \begin{equation}
      \label{eq:detPDEh0f}
      \ddt g_t = A g_t,\; t\geq 0\,\qquad  g_0 = h_0 -z_0 f.
    \end{equation}
    Then, $u_t:= g_t\mathcal{E}_t(X) +  f Z_t$ is a solution
    of~\eqref{eq:gSPDEinh} with initial condition $h_0$. 
  \item\label{inhomo:i:2factor-nec} Let $h_0\in H$ be such that there exists a weak solution
    $u = (u_t)_{t\geq 0}$ of~\eqref{eq:gSPDEinh} with initial condition
    $h_0$. Then, $g := (u - f Z)\mathcal{E}(X)^{-1}$, is a weak solution
    of~\eqref{eq:detPDEh0f}.
\end{enumerate}
\end{thm}
\begin{rmk}
  Let $(Z_t^1)_{t\geq 0}$ and $(Z_t^2)_{t\geq 0}$ be given by
  \eqref{eq:Z} with respective initial data $z_1$, $z_2>0$, $z_1\neq
  z_2$. Then, in fact $Z_t^2 - Z_t^1 = (z_2-z_1) \mathcal{E}_t(X)$,
  which is consistent with choosing different values for $z_0$ in  \ref{inhomo:i:2factor}. 
\end{rmk}
\begin{proof}
 Part~\ref{inhomo:i:1factor} follows by direct a computation: Let $\varphi \in H$, then for $t\geq 0$,
  \begin{multline}
    \label{eq:15}
    d\iprod{u_t}{\varphi}{H} = \iprod{f}{\varphi}{H} dZ_t = \\
    =\iprod{f}{\varphi}{H}\left(\lambda Z_{t-} + \alpha\right)dt +
     \iprod{f}{\varphi}{H} Z_{t-} d X_t\\
     =\left[\iprod{u_{t-}}{A^*\varphi}{H}  + \alpha \iprod{f}{\phi}{H} \right] dt + \iprod{u_{t-}}{\phi}{H} dX_t.
  \end{multline}  
  Similarly, we obtain that any solution $u$ of~\eqref{eq:gSPDEinh}
  with initial data $h_0\in H$ can be written as
  \[ u = u^{\circ, (h_0-z_0f)} + u^{(z_0f)}\]
  where $u^{\circ,(h_0-z_0 f)}$ is the solution of the
  homogeneous problem~\eqref{eq:gSPDE} with initial data $h_0-z_0 f$
  and $u^{(z_0f)}$ is a solution of~\eqref{eq:gSPDEinh} with initial data
  $z_0f$. Then, part (i) and Theorem~\ref{thm:homo} finish the proof of
  (ii) and (iii). 
\end{proof}
  It is then readily verified using It\^o's formula that the unique solution of~\eqref{eq:Z} is given by
  \begin{equation}
    \label{eq:Zexplicit}
    Z_t:=  \mathcal{E}_t(X)e^{\lambda t}\left( Z_0 + \alpha \int_0^t e^{-\lambda s} \mathcal{E}_{s-}(Y) \d s\right),\qquad t\geq 0.
  \end{equation}
  where   \[ Y_t:= -X_t + [X,X]^c_t + \sum_{s\leq t} \frac{\Delta X_s^2}{1+
      \Delta X_s},\qquad t\geq 0.\]

We now focus on the case $X = \sigma W$ for a real Brownian
motion $W$ and a constant $\sigma>0$. Then, we will consider {\it regular} two-dimensional realizations of the form $u_t = \Phi(t,Y_t)$, where
\begin{enumerate}[label=(\alph*)]
\item $Y$ is a diffusion process with state space $J \subseteq \R$, satisfying
  \[\d Y_t = b(Y_t) \d t + a(Y_t) \d W_t,\]
  for  measurable functions $b$, $a\colon J\to \R$, where $J$ has non-empty interior, $a(y) > 0$ for all $y \in J$ and $1/a$ is locally integrable on $J$.
\item $\Phi \colon [0,\infty) \times J \to  \dom(A)$ such that for
  all $\varphi \in \dom(A^*)$, the maps defined by
  $\Phi^\varphi(t,y) := \langle \Phi(t,y), \varphi\rangle$, $t\geq 0$, $y\in J$, are in $C^{1,2}(\R_{\geq 0}\times J; \R)$.
\end{enumerate}

Examples of such regular two-dimensional realizations are given by
Theorem~\ref{thm:inhomo}.(i).

\begin{thm}\label{thm:inhomo_fdr}
  Let $X_t = \sigma W_t$, $t\geq 0$, for $\sigma>0$ and a real Brownian
  motion $W$, and assume that \eqref{eq:gSPDEinh} admits a regular finite-dimensional realization $u_t= \Phi(t,Y_t)$, $t\geq 0$. Then $f$ is an eigenfunction of $A$ 
for some eigenvalue $\lambda \in\mathbb{R}$, and there exists  
an invertible transformation $h\colon J\to \R_+$  such that for $t\geq
0$, almost surely
$$ Z_t = h(Y_t), \qquad u_t = \Phi(t,h^{-1}(Z_t)) = f Z_t, $$
where $Z$ is given by \eqref{eq:Z}.
\end{thm}
\begin{proof}
  Let $\varphi \in \dom(A^*)$, and 
  \begin{equation}
    \label{eq:18}
    \Phi^\varphi(t,Y_t) := \iprod{\Phi(t,Y_t)}{\varphi}{}.
  \end{equation}
  An application of the It\^{o} formula yields
  \begin{multline}\label{eq:dux}
    \d \iprod{u_t}{\varphi}{} = d\Phi^\varphi(t,Y_t)=\\
    = \left(\partial_t \Phi^\varphi(t,Y_t) + \partial_y
      b(Y_t)\Phi^\varphi(t,Y_t)  + \frac{1}{2}a^2(Y_t) \partial_{yy}
      \Phi^\varphi(t,Y_t)  \right) dt + \\ 
    + a(Y_t)\partial_y \Phi^\varphi(t,Y_t)  dW_t.
  \end{multline}
  Comparing the martingale term  with \eqref{eq:gSPDEinh}, we see that $\Phi^\varphi$ satisfies the ODE
  \[\partial_y \Phi^\varphi(t,Y_t) = \frac{\sigma \Phi^\varphi(t,Y_t)}{a(Y_t)},\]
  $dt \otimes d\PP$-a.\,e. and hence, $\Phi^\varphi$ must be of the form 
  \[\Phi^\varphi(t,y) = g^\varphi(t) h(y) = g^\varphi(t) \exp\left(\int_{y_0}^y \frac{\sigma d\eta}{a(\eta)}\right),\qquad t\geq 0,\, y\in J,\]
  for some $g^\varphi \in C^{1}(\R_{\geq 0})$ and $y_0$ in the
  interior of $J$. The regularity property of the representation
  guarantees that $h$ is well-defined and strictly monotone
  increasing. We stress that $h$ is in fact independent of $\varphi
  \in \dom(A^*)$.
  Setting $Z_t = h(Y_t)$, we see that $Z$ satisfies
  \[dZ_t =  m(Z_t) dt  +   \sigma Z_t dW_t\]
  for the drift function $m = (b h') \circ h^{-1} + \frac{1}{2}(a^2
  h'') \circ h^{-1}$. 
  
  Note that for each $t\geq 0$, the mapping
  $\varphi \mapsto g^\varphi(t)$ is linear continuous from
  $\dom(A^*)\subset H$ into $\R$. Since $\dom(A^*)\subset H$ is dense,
  by Riesz representation theorem for each $t\geq 0$ there exists $g(t)\in H$ such that
  \begin{equation}
    \label{eq:17}
    \iprod{g(t)}{\varphi}{} = g^\varphi(t). 
  \end{equation}
  Since $\Phi^\varphi(t,y) = g^\varphi(t) h(y)$, $g^\varphi$ is differentiable and \eqref{eq:dux} becomes, for $\varphi \in \dom(A^{*})$,
  \[d\iprod{u_t}{\varphi}{} = \left( Z_t \partial_t g^\varphi(t) + g^\varphi(t) m(Z_t) \right) dt  + g^\varphi(t) Z_t dW_t.\] 
  Comparing the drift terms with \eqref{eq:gSPDEinh} yields for $t \geq 0$, $\varphi \in \dom(A^*)$ and $z \in h(J)$,
  \begin{equation}\label{eq:redux0}
    z\left( \iprod{g(t)}{A^*\varphi}{} -  \partial_t g^\varphi(t)\right) + \alpha\iprod{f}{\varphi}{} = m(z) g^\varphi(t).
  \end{equation}
  Evaluating at two different points $z_0, z_1 \in h(J)$ and subtracting we obtain that 
  \[( \iprod{g(t)}{A^*\varphi}{} - \partial_t \iprod{g(t)}{\varphi}{}) \cdot (z_1 - z_0) = \iprod{g(t)}{\varphi}{}  \cdot (m(z_1) - m(z_0)),\]
  for all $t \in \R_{\geq 0}$, $\varphi \in \dom(A^*)$ and $z_0,z_1 \in h(J)$.
  We conclude that there exists a constant $\lambda \in \R$ such that
  \begin{align}
    \iprod{g(t)}{A^* \varphi}{} - \partial_t \iprod{g(t)}{\varphi}{} &= \lambda\iprod{g(t)}{\varphi}{}\label{eq:Alambda}\\
    \intertext{and}
    m(z_1) - m(z_0) &= \lambda (z_1 - z_0).\label{eq:az}
  \end{align}
  Thus $m$ must be of the form $m(z) = \lambda z + c$ for $c:= m(0)$. Inserting into \eqref{eq:redux0} we obtain that
  \[\alpha \iprod{f}{\varphi}{} = c \iprod{g(t)}{\varphi}{},\qquad
    \forall \varphi \in \dom(A^*).\]
  Since $\dom(A^*)\subset H$ is dense the equation holds  for all $\varphi \in H$.  Due to the
  assumption that $\alpha \neq 0$ and  $f$ is non-zero, also $c \neq
  0$ and we get $g(t) =  \frac{\alpha}{c} f$. In particular, $g(t)$ is
  independent of $t$ and \eqref{eq:Alambda} yields 
  \begin{equation*}
    \langle f, A^*\varphi \rangle =  \left\langle \tfrac{\lambda c}{\alpha}
    f,\varphi\right\rangle \qquad \forall \varphi \in \dom(A^*).
  \end{equation*}
  This means that $f\in \dom(A^{**})$. Since $A=A^{**}$, see e.\,g. \cite[Theorem
  VII.2.3]{yosida1995functional}, we have $f\in \dom(A)$ and
  \begin{equation*}
    \langle f, A^* \varphi \rangle = \langle A f,\varphi \rangle = \lambda  \langle
    f,\varphi\rangle, \qquad \forall \varphi \in \dom(A^*).
  \end{equation*}
  By density of $\dom(A^*)$ in $H$ this yields that $Af = \lambda f$, i.e. $f$ must be an eigenfunction of $A$ with eigenvalue $\lambda$. 

  Putting everything together, we have shown that $u_t = \frac{\alpha}{c} f Z_t$ where 
  \[dZ_t = \left(\lambda Z_t  + c \right) dt +  \sigma Z_t dW_t.\]
  Rescaling $Z$ by $\frac{\alpha}{c}$ concludes the proof.
\end{proof}

\subsection{Linear SDEs \& Pearson diffusions}\label{sec:pearson}
Let again $X_t = \sigma W_t$ for some $\sigma>0$ and a real Brownian
motion $W$. The factor processes $Z$ appearing above are then special cases of the linear SDE
\begin{equation}\label{eq:linearSDE}
  dZ_t = (aZ_t + c) dt + (bZ_t + d) dW_t, \quad t\geq 0, \quad Z_0 = z_0,
\end{equation}
studied e.g. in \cite[Ch.~4]{Kloeden1992} or
\cite[Prop. 21.2]{kallenberg}. Well-known special cases are the
geometric Brownian motion ($c=d=0$) and the Ornstein-Uhlenbeck-process
($b=0$). Relevant in our context is the less common case $d=0$, on
which we focus now. Using ~\eqref{eq:Zexplicit}, the solution is given by
\begin{equation}
  Z_t = X_t \left(Z_0 + c\int_0^t X_s^{-1} ds\right),\quad t\geq 0,
\end{equation}
where
\begin{equation}
  X_t  = \exp\left((a -
    \frac{b^2}{2})t + bW_t\right), \quad t\geq 0. \label{eq:gBM_toinv}
\end{equation}
Solutions of \eqref{eq:linearSDE} have also been studied in the
context of reciprocal gamma diffusions (see e.g. the `Case~4' in
\cite{forman2008pearson}) or also Pearson diffusions. These are generalizations of
\eqref{eq:linearSDE} that allow for a square-root term in the
diffusion coefficient. 

\begin{prop}\label{prop:pearson_ergodicity}
  Assume that $z_0>0$, $a<0$ and $c>0$. Then, $Z$ has   unique invariant distribution
  $\varpi$, which  is an Inverse Gamma distribution with shape parameter $1-\frac{2a}{b^2}$ and scale
  parameter $\frac{b^2}{2c}$ and, for any bounded  measurable function
  $\phi \colon (0,\infty)\to \R$,
  \[\lim_{t \to \infty} \E{\phi(Z_t)} = \lim_{t\to\infty} \frac1{t} \int_0^t \phi(Z_s)\d s=  \int_0^\infty \phi(x)\varpi(\d
    x).\]
\end{prop}
\begin{proof}
  First, note that
  \[ s'(x) := x^{-2\frac{a}{b^2}} e^{2 \frac{c}{b^2 x}},\qquad m(\d x):=
    x^{2(\frac{a}{b^2} - 1)} e^{-2\frac{c}{b^2 x}} \d x,\quad x\in
    (0,\infty)\]
  define a scale density and speed measure for $Z$. Then, one can easily
  verify that $Z$ is strictly positive and recurrent on $(0,\infty)$, see
  e.\,g. \cite[Prop. 5.5.22]{karatzas2012brownian}. Moreover, $m((0,\infty))
  <\infty$ and so the unique invariant distribution of $Z$ is
  \begin{equation}
    \label{eq:7}
    \varpi(A) := \frac{m(A)}{m((0,\infty))}. 
  \end{equation}
  The remaining results then follow from e.\,g. \cite[II.35]{borodin2012handbook} or
  \cite[X.3.12]{ry}.
\end{proof}
 $\mu(t) := \EE Z_t$, $t\geq 0$, satsifies the ODE
\begin{equation*}
  \ddt \mu(t) = a \mu(t) + c,\; t>0,\;\qquad \mu(0) = Z_0.
\end{equation*}
Thus,
\begin{equation}
  \label{eq:meanPearson}
  \mu(t) = \left(Z_0 + \frac{c}{a} \right)e^{a t}  -\frac{c}{a}.
\end{equation}
\begin{rmk}\label{rmk:autocorr}
  Let $a<0$, $c>0$ and $(Z_t)$ be the stationary solution of
  \[ \d Z_t = \left( aZ_t + c\right) \d t + b Z_t\d W_t,\]
  that is, $Z_0$ is chosen distributed according to inverse gamma distribution with shape parameter $1-\frac{2a}{b^2}$ and scale parameter $\frac{b^2}{2c}$. Then, as shown in \cite{bibby2005diffusion}, the autocorrelation function of $(Z_t)$ is given by
  \begin{equation}
    \label{eq:autocorr}
    r(t) := \operatorname{Corr}(Z_{s+t}, Z_s) = e^{at},\qquad s,\,t\geq 0. 
  \end{equation}
\end{rmk}
To study price dynamics it is  also useful to examine  the reciprocal process $Y=1/Z$. When 
$d=0$, $Y = 1/Z$ is the unique solution of 
  \begin{equation}
    \label{eq:logistic}
    \d Y_t =- Y_t (a -b^2 +  c Y_t) \d t - b Y_t \d W_t,\qquad Y_0 = z_0^{-1}.
  \end{equation}
  In particular, with $X$ given in~\eqref{eq:gBM_toinv},
  \begin{equation}
    \label{eq:6}
    Y_t = \mathcal{E}_t(-b W_. -a (.)) \left(Z_0 + c \int_0^tX_s^{-1}\d s\right)^{-1},\quad t\geq 0.
  \end{equation}
When $a<b^2$, ~\eqref{eq:logistic} is called the stochastic logistic  equation.

\subsection{Positivity, stationarity and martingale property}\label{sec.martingale}
Let us first come back to the linear homogeneous situation. On
average, market makers do not accumulate inventory, which suggests to consider the baseline case of {\it balanced order flow}for  which  $X$ is a (local)
martingale. 
 If $X $ is a local martingale with $\Delta X>-1$ a.\,s., then, from the properties of stochastic exponentials we obtain that:
\begin{itemize}
\item   The    weak solution $u_t$ of the homogeneous equation \eqref{eq:gSPDE}  is a local martingale,
    if and only if the initial condition $h_0$ is $A-$harmonic: $h_0\in \dom(A)$ and $Ah_0 = 0$. 
  \item  If $\mathcal{E}(M)$ is a  martingale and $Ah_0 = 0$, then $(u_t)_{t\geq 0}$ is a martingale.
  \end{itemize}

In the Brownian motion case, from the discussion in the previous section we directly obtain: 
\begin{cor}\label{eq:stationarity_inh}
  Let $X= \sigma W$ where $W$ is  a standard Brownian motion   and $\sigma >0$, and  $u$ be the solution of the inhomogeneous equation \eqref{eq:gSPDEinh}, where $f$ is an eigenfunction of $A$ with eigenvalue $-\nu$ and $h_0 = z_0 f$, for some $z_0>0$. If $\nu > 0$ and $\alpha > 0$, then 
  $$u_t \mathop{\Rightarrow}^{t\to \infty}\quad f Z_\infty$$ where $Z_\infty$ has an Inverse Gamma distribution with shape parameter $1 + 2 \frac{\nu}{\sigma^2}$ and scale parameter $\frac{\sigma^2}{2\alpha}$.
\end{cor}
\begin{rmk}\label{rmk:invgamma}
  The Inverse Gamma distribution has a Pareto (right) tail with tail index  $1 + 2 \nu$ in this case: the $k$-th moment of $\mathbb{E}(Z_\infty^k)<\infty$  if and only if $k < 1 + 2\nu$.
\end{rmk}

So far, we have set aside  the positivity constraint for $u$. By Theorem~\ref{thm:homo} this reduces to analysis of the deterministic equation. In the case of second-order elliptic operators, positivity results from the comparison principle, whenever the initial condition $h_0$ is positive:
\begin{ass}\label{ass:elliptic}
  Let $I\subset \R$ be an interval and suppose that $A$ is a uniformly elliptic operator of the form 
  \[Au(x) = \eta(x) \Delta u(x) + \beta(x) \nabla u(x) + \alpha(x) u(x), \qquad x\in I,\]
  with Dirichlet boundary conditions, and where $\eta, \beta$ and $\alpha$ are smooth and bounded coefficients, and in particular $\eta(x) \geq \underline{\eta} >0$ for all $x\in I$.
\end{ass}

In addition, the principal eigenvalue of $A$,  $\lambda_1$ has  an eigenfunction $f$  which is positive on $I$  \cite[Sec.~6.5]{Evans2010}. Note that the factor process $Z_t$ has state space $(0,\infty)$ both in Theorem~\ref{thm:homo} and \ref{thm:homo}. We thus obtain the following corollary.
\begin{cor}[Positivity]
  Under Assumption~\ref{ass:elliptic},
  \begin{enumerate}[label=(\roman*)]
  \item If $h_0$ is positive on $I$, then the solution $g_t$ of \eqref{eq:gPDE} and the solution $u_t$ of \eqref{eq:gSPDE} are a.s. positive on $I$.
  \item If $f$ is the principal eigenfunction of $A$, then the finite-dimensional realization $u_t = f Z_t$ of \eqref{eq:gSPDEinh} is a.s. positive on $I$.
  \end{enumerate} 
\end{cor}
This simple result thus guarantees the existence of a solution with the correct sign, thereby avoiding recourse to `reflected' solutions as in \cite{hambly2018} and considerably simplifying the analysis of our model.

\section{A two-factor model}
\label{sec:2factorLOB}
We now study the simplest example of  model  satisfying Assumption~\ref{ass:elliptic}, namely the case of constant coefficients $\eta_a$, $\eta_b$, $\sigma_a$, $\sigma_b >0$, $\beta_a$, $\beta_b \geq 0$, $\alpha_a$, $\alpha_b\in \R$;
\begin{equation}
  \begin{split}
    \d u_t(x) &= \left[\eta_a \Delta u_t(x) + \beta_a \nabla u_t(x) + \alpha_a u_t(x) \right]\d t 
    + \sigma_a u_t(x)  \d W^a_t ,  \quad\,x \in (0,L), \\
    \d u_t(x) &= \left[\eta_b \Delta u_t(x) - \beta_b \nabla u_t(x) + \alpha_b u_t(x) \right]\d t 
    + \sigma_b u_t(x)  \d W^b_t,  \quad\,x \in (-L,0),\\    
    u_t(x) &= 0,\quad x \in\{-L,0,L\}, \qquad
    u_0  \in L^2(-L,L). \qquad 
  \end{split} \label{eq:lSPDEhom}
\end{equation}
together with the sign condition:
\begin{equation}
  \label{eq:positivity}
  u_t(x) \leq 0, \quad x\in (-L,0),\quad \text{and}\quad u_t(x) \geq 0,\quad x\in (0,L), \;t\geq 0.
\end{equation}
In the following, we will write $u^b_0 := u_0\vert_{[-L,0]}$ and $u^a_0 := u_0\vert_{[0,L]}$.

\subsection{Spectral representation of solutions}

A spectral representation of the operator may be used to obtain an analytical solution to this model.
\begin{prop}\label{prop:ev}
  Let $I =(-L,0)$ or $I= (0,L)$ and $\eta >0$, $\beta$, $\alpha \in
  \R$, and consider the linear operator 
  \begin{equation}
    A := \eta \Delta + \beta \nabla + \alpha \Id  
  \end{equation}
  on $L^2(I)$, with $\dom(A) := \left\{u\in H^2(I)\vert\, u\vert_{\partial I} =
    0\right\}=H^2(I)\cap H^1_0(I)$.
  The eigenvalues of $-A$ are real and given by
\begin{equation} \nu_k =  -\alpha +  \frac{\eta k^2 \pi^2}{L^2} + \frac{\beta^2}{4\eta},\quad k= 1, 2, ... \label{eq:spectrum}\end{equation}
  with corresponding eigenfunctions
  \[h_k(x) := e^{-\frac{\beta}{2\eta} x} \sin\left(\tfrac{k\pi }{L} x\right),\qquad x\in I.\]
  In particular  the only positive eigenfunction is $h_1$.
\end{prop}

\begin{proof}
  First we note that that $\phi$ is an eigenfunction of $A$ with eigenvalue $\nu$, if and only if 
  \[x \mapsto e^{\frac{\beta}{2\eta} x} \phi(x)\]
  is an eigenfunction of $A_0:= \eta \Delta + \alpha \Id$ with zero Dirichlet boundary conditions, for
  eigenvalue $\nu + \frac{\beta^2}{4\eta}$. Details of calculations are given in \cite{cont2005}.
  The operator $A_0$ with
  domain $\dom(A_0):= \dom(A)$ is self-adjoint, has compact resolvent \cite{cont2005} and eigenvalues
  \begin{equation}
    \label{eq:2}
    \alpha -  \frac{\eta k^2 \pi^2}{L^2},\qquad k\in \N.
  \end{equation}
  Eigenfunctions of $A_0$ with eigenvalue  $\nu\in \R$ are solutions of  the Sturm-Liouville problem
  \begin{equation}
    \label{eq:EVode0}
    \eta g''(x)  + (\alpha - \nu) g(x) = 0,\qquad x\in I,
  \end{equation}
with zero boundary conditions, which 
  yields that $g$ must be of the form
  \begin{equation}
    \label{eq:15}
    g(x) = c e^{-\gamma_1 x} \sin(\gamma_2 x),\quad{\rm where}\quad \gamma_1 = 0,\qquad \gamma_2 = \frac{\nu- \alpha}{\eta}.
  \end{equation}
 The zero boundary conditions at $0$ and $\pm L$  imply 
  $\gamma_2 = \frac{k}{L}\pi$ for some $k\in
  \N$ so
  \begin{equation}
    \label{eq:gBMdrift}
    \nu = \alpha -  \frac{\eta k^2 \pi^2}{L^2}.  
  \end{equation}
  Translating this from $A_0$ to $A$  yields the result.
\end{proof}
Define the following  bilinear forms:
\begin{equation}
  \label{eq:iprod_gamma-}
  L^2(-L,0)\times L^2(-L,0) \ni (f,g) \mapsto \iprod{f}{g}{-\gamma} := \frac{2}{L}\int_{-L}^0 f(x) g(x) e^{-2\gamma x} \d x
\end{equation}
and 
\begin{equation}
  \label{eq:iprod_gamma+}
  L^2(0,L)\times L^2(0,L) \ni (f,g) \mapsto \iprod{f}{g}{\gamma} := \frac{2}{L}\int_0^L f(x) g(x) e^{2\gamma x} \d x
\end{equation}
which define equivalent inner products, respectively for $L^2(-L,0)$
and $L^2(0,L)$. For $\gamma > 0$, and $k\in\N$, define
\begin{align}
  \nu_k^a :=  -\alpha_a +&  \frac{\eta_a k^2 \pi^2}{L^2} +
                           \frac{\beta_a^2}{4\eta_a},\quad
                           \nu_k^b :=  -\alpha_b +  \frac{\eta_b k^2 \pi^2}{L^2} +
                           \frac{\beta_b^2}{4\eta_b},\\    
  h_k^a(x) &:= e^{\frac{-\beta_a}{2\eta_a}x}
             \sin\left(\frac{k\pi}{L}x\right),\qquad x\in (0,L), \label{eq.hka}\\
  h_k^b(x) &:=  e^{\frac{\beta_b}{2\eta_b}x}
             \sin\left(\frac{k\pi}{L}x\right),\qquad x\in (-L,0). \label{eq.hkb}
\end{align}
Let
\begin{equation}
  \label{eq:1}
   \gamma_a := \frac{\beta_a}{2\eta_a},\qquad \gamma_b := \frac{\beta_b}{2\eta_b}.
\end{equation}
 Then $(h_k^b)_{k\in \N}$ is an orthonormal basis of $\left(L^2(-L,0),\ip{\cdot}{\cdot}{-\gamma_b}\right)$ and $(h_k^a)_{k\in \N}$ is an orthonormal basis for 
$ \left(L^2(0,L),\ip{\cdot}{\cdot}{\gamma_a}\right)$ and solutions for the SPDE may be constructed using an expansion along these bases:
\begin{prop}
  \label{prop:solutionLinHom}
  Let $u_0\in L^2(-L,L)$, $u_0^a:=u_0\vert_{[0,L]}$, $u_0^b:= u_0\vert_{[-L,0]}$.Then $(u_t)_{t\geq 0}$ defined by
  \begin{equation}
    u_t(x) :=
    \begin{cases}
      \mathcal{E}_t(\sigma_b W^b)\sum_{k=1}^\infty e^{-\nu_k^b t}
      \iprod{u_0^b}{h_k^b}{-\gamma_b} h_k^b(x),&\qquad x\in
      (-L,0), \\
      \mathcal{E}_t(\sigma_a W^a)\sum_{k=1}^\infty e^{-\nu_k^a t}
      \iprod{u_0^a}{h_k^a}{\gamma_a} h_k^a(x),&\qquad x\in
      (0,L), \\
      0,& x\in \{-L,0,L\}.
    \end{cases}          
  \end{equation}
is the unique continuous weak solution
  of~\eqref{eq:lSPDEhom} in the sense of Definition~\ref{df:weak_solution}.
\end{prop}
\begin{proof}
  The unique continuous solutions of the respective deterministic equations are given by
  $(S^b_t u^b_0)_{t\geq 0}$ and $(S^a_t u^a_0)_{t\geq 0}$, where $(S_t^b)_{t\geq 0}$ and
  $(S_t^a)_{t\geq 0}$ are the Dirichlet semigroups generated by 
  \begin{equation}
    A_b = \eta_b \Delta u_t - \beta_b \nabla  + \alpha_b
    \quad\text{and}\quad   A_a =\eta_a \Delta + \beta_a \nabla  + \alpha_a 
  \end{equation}
  on $(-L,0)$ and $(0,L)$, respectively. Thus, from
  Theorem~\ref{thm:homo} we get 
  \begin{equation}\label{eq:lSPDEhomSolAbstract}
    u_t(x) =
    \begin{cases}
      \mathcal{E}_t(\sigma_b W^b)S_t^b u_0^b(x),& x \in (-L,0),\\
      \mathcal{E}_t(\sigma_a W^a)S_t^a u_0^a(x),& x \in (0,L).
    \end{cases}
  \end{equation}
  $(S_t^a)$ and $(S_t^b)$ are linear continuous so
  that for each $h^a\in L^2(0,L)$, $h^b \in L^2(-L,0)$,
  \[ S_t^a h^a = \sum_{k\in
      \N}\iprod{u^a_0}{h_k^a}{\gamma_a}
    S_t^a h_k^b,\qquad{\rm
  and }
  \qquad S_t^b h^b = \sum_{k\in
      \N}\iprod{u^b_0}{h_k^b}{-\gamma_b}
    S_t^b h_k^b.\]
  By Proposition~\ref{prop:ev} $h_k^a$ (resp. $h_k^b$) are eigenfunctions of
  $A_a$ (resp. $A_b$) and thus also of $S^a$ (resp. $S^b$). This yields
  the desired representation, where the  series converge in
  $L^2$. To obtain pointwise convergence, we note that for $x\in [0,L]$ and $t>0$,
  by Cauchy-Schwarz inequality, Parseval's identity and integral
  criterion for sequences, for $\star \in \{a,b\}$,
  \begin{align*}
    \sum_{k=1}^\infty \abs{e^{-\nu_k^\star t}
      \iprod{u_0\vert_{(0,L)}}{h_k^\star}{\frac{\beta_\star}{\eta_\star}} h_k^\star(x)}
    &\leq \norm{h\vert_{(0,L)}}{L^2(0,L)}
      \sqrt{\sum_{k=1}^\infty e^{-2\nu_k^\star t}}\\
    &\leq  \norm{u_0\vert_{(0,L)}}{L^2(0,L)}^2
      e^{t(\alpha_\star - \frac{\beta_\star^2}{4\eta_\star})} \sqrt{\int_0^\infty
    e^{-2t\frac{\eta_\star\pi^2}{L^2}y^2} \d y}.\qedhere
  \end{align*}
\end{proof}
When $\eta>0$, then the weights $e^{-\nu_k t}$ of the spectral
decomposition decay exponentially in $k^2$ for large $k$. This 
justifies approximating the solution by  the  first few
terms. Note also that the only positive eigenfunctions are the principal eigenfunctions $h^a_1$ and $-h^1_b$ so the sign
constraints~\eqref{eq:positivity} only if the projection of the solution along the principal eigenfunctions dominates the other terms in the expansion. 
This motivates us  to focus on solutions which live in the first eigenspace. This occurs  if the initial condition is a (positive) linear combination of  $h^a_1$ and $h^1_b$. 
We will later show that this assumption is supported by  market data. 
This leads to a finite-dimensional realization which satisfies the sign constraints~\eqref{eq:positivity}:
\begin{cor}\label{cor:fdr}
  Let $V_0^a>0$ resp. $V_0^b>0$ and define
    \begin{equation} \label{eq:L1} H_1^a(x)=\frac{h^a_1(x)\1_{(0,L)}(x)}{\int_0^L |h^a_1|}\geq 0\qquad{\rm and}  \qquad H_1^b(x)=\frac{h^b_1(x) \1_{(-L,0)}(x)}{\int_{-L}^0 |h^b_1|}\leq 0.\end{equation} 
  The  unique solution of~\eqref{eq:lSPDEhom}--\eqref{eq:positivity} with initial condition $u_0= V_0^a H_1^a +  V_0^b H_1^b$  is given by
  \begin{equation}
    \label{eq:8}
    u_t(x) = H_1^b(x)V^b_t  + H_1^a(x)V^a_t ,\quad t\geq 0,\, x \in [-L,L],\quad {\rm where}
  \end{equation} 
  \begin{align}\label{eq.nu}
    \nu^a = -\alpha_a +  \frac{\eta_a \pi^2}{L^2} + \frac{(\beta_a)^2}{4\eta_a}, \qquad 
    \nu^b = -\alpha_b +  \frac{\eta_b \pi^2}{L^2} + \frac{(\beta_b)^2}{4\eta_b}\quad {\rm and}
  \end{align}
  \begin{alignat}{2}
    \label{eq:16}
    \d V_t^a = -\nu^a V_t^a \d t+ \sigma_a V_t^a \d W_t^a,& \quad
    \d V_t^b = -\nu^b V_t^b \d t+ \sigma_b V_t^b \d W_t^b
  \end{alignat}
  In particular, $u_t\vert_{[-L,0]}\leq 0$, $u_t\vert_{[0,L]}\geq 0$ and 
  \[  \nabla u_t(0+) = \frac{\pi}{L} V_t^a ,\quad  \nabla u_t(0-) =  \frac{\pi}{L}  V_t^b.\]  
\end{cor}
The $L^1$ normalization \eqref{eq:L1} allows to interpret the variables  in terms of order book volume and depth:
$\int_0^L|u_t|=V_t^a$ (resp. $\int_{-L}^0|u_t|=V_t^b$) represents the volume of sell (resp. buy) orders, while $\nabla u_t(0+)\theta = \frac{\theta \pi}{L} V_t^a $ (resp. $ \nabla u_t(0-) .\theta=  \frac{\theta \pi}{L}  V_t^b$) represents the depth at the top of the book.
In this simple two-factor model, these two  are proportional to each other: they may be decoupled by considering multifactor specifications involving higher-order eigenfunctions. 

The drift parameter $-\nu^a$ (resp. $-\nu^b$)  thus represents the net growth rate of decrease of the volume of sell (resp. buy) orders. As shown in \eqref{eq.nu}, this net growth rate results from the superposition of several effects:
\begin{itemize}
\item submission/ cancellation  of limit sell (resp. buy) orders by directional sellers (resp. buyers) at rate $\alpha_a$ (resp. $\alpha_b$); this may be interpreted as the `low frequency' component of the order flow;
\item replacement of limit orders by new ones closer to the mid-price, at rate 
$\frac{\beta_a^2}{4\eta_a}$ (resp. $\frac{\beta_b^2}{4\eta_b}$);
\item cancellation of limit orders as the mid-price moves away (i.e. at distance $\pm L$ from the mid-price), at rate $\frac{\eta_a \pi^2}{L^2}$ (resp. $\frac{\eta_b \pi^2}{L^2}$).
\end{itemize}
In the case of a {\it balanced} order flow for which 
there is no systematic accumulation or depletion of  limit orders away from the mid-price, these terms compensate each other  and
the volume of limit orders in any interval $[S_t+x_1,S_t+x_2]$ is a (local) martingale.
The following result follows from the remarks in Section~\ref{sec.martingale}:
\begin{cor}[Balanced order flow]
  The order book density $u$ is a  local martingale (in $L^2$), if and only if 
  \[u_0(x) =  V^b_0 H_1^b(x) \1_{(-L,0)}(x) + V^a_0 H_1^a(x) \1_{(0,L)}(x),\]
  for some $V^b_0 \geq 0,V^a_0 \geq 0$ and
  \begin{equation}
    \label{eq:3}
    \alpha_{ a} = \frac{\eta_{a} \pi^2}{L^2} + \frac{\beta_{a}^2}{4\eta_{ s}},
    \quad
    \alpha_{ b} = \frac{\eta_{ b} \pi^2}{L^2} + \frac{\beta_{b}^2}{4\eta_{ b}}.
  \end{equation}
\end{cor}

\begin{rmk}[Balance between high- and low-frequency order flow]
The balance condition \eqref{eq:3} expresses a balance between the slow arrival of directional orders, represented by the terms $\alpha_a$ and $\alpha_b$, and the fast replacement of orders inside the book, represented by the terms $\frac{\beta_{a}^2}{4\eta_{ b}}$ and  $\frac{\beta_{a}^2}{4\eta_{ b}}$, and finally the cancellation of limit orders deep inside the book, at rate ${\eta_{a} \pi^2}/{L^2}.$
\end{rmk}
This balance between order flow at various frequencies may be seen as  a mathematical counterpart of the observations  made by \cite{kirilenko2017} on the nature of intraday order flow.

\subsection{Shape of the order book}
\label{ssec:obshape}
An implication of the above results is that the average profile of the order book is given, up to a constant, by the principal eigenfunctions $H_1^a, H_1^b$:
\begin{equation} \mathbb{E}( u_t(x) )= \mathbb{E}(V^b_t)\  H_1^b(x)  +   \mathbb{E}(V^a_t )\ H_1^a(x) \label{eq:average}\end{equation}
Dropping the indices $a,b$, the normalized profile of the order book has the  form:
\[ H_1(x) := c_1 e^{-\frac{\beta}{2\eta} x}\sin(\tfrac{\pi}{L}x),\qquad x \in [0,L],\]
where $c_1$ is such that $\int_0^L |H_1|=1$:
\[ \frac{1}{c_1}=  \int_0^L e^{-\frac{\beta}{2\eta} x}\sin(\tfrac{\pi}{L}x) \d x = \frac{4\pi L\eta^2}{L^2\beta^2 + \pi^2 4\eta^2}\left(e^{-\frac{\beta}{2\eta}L}+1\right),\]
Figure   \ref{fig:shapes} shows this function for different values of $\beta$:
$H_1$ has a unique maximum  at
\begin{equation}\label{eq:xhat}
\hat x :=  \frac{L}{\pi} \arctan\left(\tfrac{2 \eta \pi}{L\beta}\right). 
\end{equation}
The position of the maximum moves closer to the origin as $\sfrac{\beta}{\eta}$ is increased.
For $\beta = 0$ we have $\hat x = \frac{L}2$, and, on the other hand
$\hat x \searrow 0$ as $\sfrac{\beta}{\eta}\to \infty$. 
Typically, the order book profile for liquid large--tick securities a few ticks from the mid price. Figure~\ref{fig:profile_fit_qqq} shows the average order book  profile for QQQ; similar results were found  in  \cite{bouchaud2009,cont2010}. 
This suggests $\hat x$ is of the order of a few ticks, so we are interested in the parameter range for which $\sfrac{\beta}{\eta}$ is large.

The value at the maximum is
\begin{multline}
  \label{eq:maxvolratio}
 \max_{x \in [0,L]} H_1(x)  
  = \sqrt{ \frac{\beta^2}{4\eta^2} +\frac{ \pi^2}{L^2}} \exp\left(-\frac{\beta L}{2\eta \pi} \arctan\left(\frac{2\eta \pi}{L\beta}\right)\right) \left(e^{-\frac{\beta L}{2\eta}} + 1 \right)^{-1}.
\end{multline}
which grows linearly as $\sfrac{\beta}{2\eta} \to \infty$, as shown in 
Figure~\ref{fig:shapes}, where we have plotted $h$, normalized by its $L^1$-norm, for various values
of $\beta$ with $L:= 3\pi$ and $\eta =1$. 

\begin{figure}[bt]
  \centering
  \includegraphics[width=\textwidth]{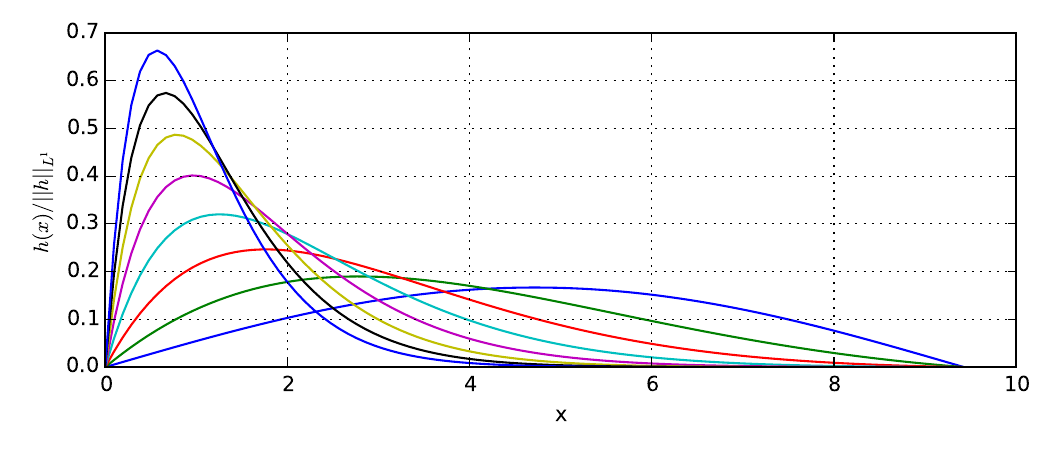}
  \caption{ Shape of the normalized principal eigenfunction $H_1$, which corresponds to the average profile of the normalized order book,
   for $L:=  3\pi$, $\eta:=1$ and different values of $\beta\in \{0,0.5,...3.5\}$.}
  \label{fig:shapes}
\end{figure}
The above results are valuable for calibrating the model parameters $\frac{\beta}{2\eta}$, $\alpha$ and $\sigma$ to reproduce the average profile (for each side) of the order book. 

  $\frac{\beta}{2\eta}$ can be estimated  from the position $\hat x$ of the maximum   using \eqref{eq:xhat}. Note that, when $L$ is large then
    \[\hat x\approx \frac{2\eta}{\beta}.\]
  The height of this maximum gives a further constraint on parameters, using   \eqref{eq:maxvolratio}.
  
    We will use this result for parameter estimation in
  Section~\ref{ssec:fitHom}.

\subsection{Dynamics of order book volume}
As noted in Corollary~\ref{cor:fdr}, $V^a_t$ and $V^b_t$ may be identified as the volume of sell (resp. buy) limit orders: they follow (correlated) geometric Brownian motions:
\begin{align*}
  \label{eq:11}
  V_t^a &= \int_{0}^L \abs{u_t(x)}\d x = V^a_0 \exp(\sigma_a W^a_t-\nu_a t-\frac{\sigma_a^2t}{2})\\
  V_t^b &= \int_{-L}^0 \abs{u_t(x)}\d x = V^b_0 \exp(\sigma_b W^b_t-\nu_b t-\frac{\sigma_b^2t}{2})
\end{align*}
where $[W^a, W^b]_t = \rho_{a,b} t$.
The average volume of the order book $V_t=V^a_t+V^b_t$ satisfies
\begin{equation*}
  \EE ( V_t )= V_0  -\int_0^t V^0_a \nu_a e^{-\nu_a s} -  V^0_b \nu_b e^{-\nu_b s} \d s= V_0 +V^0_a e^{-\nu_a t}+ V^0_b e^{-\nu_b t}.
\end{equation*}
Intraday studies of order book volume show it to be stable away from the open and close. Here
 $\EE V_t = V_a + V_b$  if and only
if $V$ is a martingale, i.\,e. $\nu_a = \nu_b = 0$. 

\subsection{Dynamics of price and market depth} 
\label{ssec:pricelinh}
Recall from the discussion in Section~\ref{sec:price} that the
order book dynamics yield the price process
\begin{equation*}
  \d S_t =   \theta \left(\frac{\d D_t^b}{D_t^b} - \frac{\d
      D_t^a}{D_t^a}\right),
\end{equation*}
where $\theta$ is an impact coefficient and
$D_t^b$ and $D_t^a$ represent the depth at the top of the order book
\cite{contImbalance}:
\begin{equation}
  \label{eq:6}
  D_t^a := \int_0^{\delta} u_t(x)\d x\approx \frac12 \delta^2 \nabla u_t(0+),\quad
  D_t^b := \int_{-\delta}^0 |u_t(x)| \d x \approx \frac12 \delta^2\nabla u_t(0-).
\end{equation}
Using the results in Corollary~\ref{cor:fdr}, we obtain the following price dynamics:
\begin{equation}
  \label{eq:pricemodel_linhom}
  \d S_t =   \theta  \left(\frac{\d V_t^b}{V_t^b} - \frac{\d V_t^a}{V_t^a}\right),
\end{equation}
where
\begin{equation}
  \d V_t^b = -\nu_b V_t^b \d t + \sigma_b  V_t^b \d W_t^b ,\quad
  \d V_t^a = -\nu_a V_t^a \d t + \sigma_a V_t^a \d W_t^a .
\end{equation}
The price dynamics can thus be written as
\begin{align*}
  S_t &= S_0 - \theta t \left(\nu_b - \nu_a\right)  +  \theta \sigma_b W_t^b -  \theta \sigma_a W_t^a \\
      &{=} S_0 - \theta t \left(\nu_b - \nu_a\right)  +
        \sigma_S B_t
\end{align*}
where $B$ is a Brownian motion and $\sigma_S$ is the mid price volatility, which may be expressed in terms of parameters describing the order flow:
\begin{equation}
  \label{eq:midprice_vol}
  \sigma_S :=   \theta  \sqrt{\sigma_b^2 + \sigma_a^2 - 2\sigma_a \sigma_b\varrho_{a,b}}.
\end{equation}
The implied price dynamics thus corresponds to the Bachelier model:
\begin{itemize}
\item The drift term $\nu_a - \nu_b$ only depends on the rate of relative increase of the bid/ask depth, not the actual depths $D^b_t$ and
  $D^a_t$. 
\item The quadratic variation of the mid price is $\sigma_S^2t$ decreases with the correlation between the buy and sell order flow. This correlation, generated by market makers, reduces price volatility.
\end{itemize}
\begin{rmk}
  Replacing $\sigma_aW^a$ and $\sigma_b W^b$ by
  arbitrary semimartingales $X^a$ and $X^b$ with jumps bounded from
  below by $-1$, yields the following price dynamics:
  \begin{equation}
    S_t = S_0 -  \theta  t  \left(\nu_b - \nu_a\right)  +   \theta  (X_t^b - X_t^a).
  \end{equation}
  In particular, this relation  links price jumps to large changes (`jumps') in order flow imbalance:
    \begin{equation}
 \Delta S_t= \theta \   \Delta (X_t^b - X_t^a).
  \end{equation}
\end{rmk}

\subsection{Absolute price coordinates: stochastic moving boundary problem} 
\label{ssec:2factor_noncentered}
The model above describes dynamics of the order book in relative price coordinates, i.e. as a function of the (scaled) distance $x$ from the mid-price. 
The density of the limit order book parameterized by the (absolute) price level 
$p\in \R$  is given (in the case of linear scaling) by
\begin{equation}
  \label{eq:1}
  v_t(p) = u_t(p - S_t),\qquad x\in \R,
\end{equation}
where we extend $u_t$ to $\R$ by setting $u_t(y) = 0$ for
 $y\in \R\setminus [-L,L]$. As observed in
Section~\ref{ssec:pricelinh}, the mid-price dynamics is given by
\begin{equation}
  \label{eq:7}
  \d S_t = -  \theta  (\nu_b - \nu_a)\d t +   \theta \sigma_b \d
  W_t^b -   \theta  \sigma_a \d W_t^a.
\end{equation}
The dynamics of $v$ may then be described, via an application of the It\^{o}-Wentzell formula, as the solution of a {\it stochastic moving boundary problem}  \cite{mueller2016stochastic}: 
\begin{thm}[Stochastic moving boundary problem]\label{thm:non-centered_hom}
  The order book density $v_t(p)$, as a function of the price level $p$ is a solution, in the sense of distributions, of the stochastic moving boundary problem
  \begin{multline}
    \label{eq:vts+}
    \d v_t(p) = \left[(\eta_a+\tfrac12 \sigma_s^2) \Delta v_t(p)\right.\\
    \left.+ (\nu_b - \nu_a + \beta_a -  \theta  (\varrho_{a,b} \sigma_b \sigma_a -  \sigma_a^2) \nabla v_t(p)  + \alpha_a v_t(p)\right] \d t \\
    + \left( \sigma_a v_t(p)  +   \theta  \sigma_a \nabla v_t(p)\right) \d W_t^a  -  \theta\sigma_b \nabla v_t(p) \d W_t^b,    
  \end{multline}
  for $p \in (S_t, S_t+L)$, and
  \begin{multline}
    \label{eq:vts-}
    \d v_t(p) = \left[(\eta_b+\tfrac12 \sigma_s^2) \Delta v_t(p)\right.\\
    \left.+ (\nu_b - \nu_a - \beta_b -   \theta (\sigma_b^2-\varrho_{a,b} \sigma_b \sigma_a)) \nabla v_t(x)  + \alpha_b v_t(p)\right] \d t \\
   +  \theta   \sigma_a \nabla v_t(p)\d W_t^a     + \left( \sigma_b v_t(p) -  \theta \sigma_b \nabla v_t(p)  \right) \d W_t^b  
  \end{multline}
  for $x\in (S_t-L, S_t)$ with the   moving  boundary conditions
  \begin{equation}
    \label{eq:bdryst_hom}
    v_t(S_t) = 0,\qquad v_t(y) = 0,\qquad \forall y\in \R\setminus(S_t-L,S_t+L),
  \end{equation}
  in the following sense:
  $(v_t)_{t\geq 0}$ is an continuous $L^2(\R)$-valued stochastic process  and for all 
  $\varphi \in C^\infty_0(\R)$ and $t\geq 0$,
  \begin{multline}
    \langle v_t, \varphi \rangle - \langle v_0, \varphi \rangle 
    =\int_0^t \langle m (x-S_t,\Delta v_r, \nabla v_r, v_r) ,\varphi
    \rangle \d r  \quad+ \\
   \frac12 
     \int_0^t \left(\nabla v_r(S_r-)-  \nabla v_r(S_{r}+) ) \varphi(S_r) - \nabla v_r(S_{r}-L+)  \varphi(S_r-L)  + \nabla v_r(S_{r}+L-)  \varphi(S_r+L)  \right)d\langle S \rangle_r 
    \\
    +\int_0^t \langle \1_{(S,S_r+L)} \sigma_a v_r, \varphi\rangle \d W_r^a
    + \int_0^t \langle \1_{(S_r-L,S_r)} \sigma_b v_r , \varphi \rangle \d W^b_r \\
    +  \theta \sigma_a  \int_0^t \langle \nabla v_r, \varphi \rangle \d W^a_r 
    -  \theta \sigma_b \int_0^t \langle \nabla v_r, \varphi \rangle \d W^b_r,
  \end{multline}
  where we denote,  for $S\in \R$, $V\in H^1_0((-L,L)\setminus\{0\}) \cap H^2((-L,L)\setminus\{0\})$,
  \begin{equation*}
    m(x, y'',y',y) =
    \begin{cases}
      (\eta_a +\tfrac12 \sigma_s^2) y'' &\\
      + (\nu_b - \nu_a + \beta_a -  \theta  (\varrho_{a,b} \sigma_b \sigma_a -  \sigma_a^2)) y' + \alpha_a y,& x\in (0,L),\\
      (\eta_b +  \theta  \sigma_s^2) y''&\\
      + (\nu_b - \nu_a - \beta_b -    \theta (\sigma_b^2-\varrho_{a,b} \sigma_b \sigma_a))) y' + \alpha_b y,& x\in (-L,0)\\
      0, & \text{ else,}
    \end{cases}
  \end{equation*}
  for $x, y'',y',y \in \R$.
\end{thm}
\begin{rmk}
  Note that  \eqref{eq:bdryst_hom} is a {\it stochastic boundary condition} at $S_t$.
\end{rmk}
The proof, given in Appendix~\ref{A:trafo},
is based on
Krylov's extended It\^{o}-Wentzell
formula \cite[Theorem 1.1]{krylovItoWentzell}.

\subsection{Parameter estimation}
\label{ssec:fitHom}

We now describe a method for estimating model parameters. We use time series of order books for NASDAQ stocks and ETFs, from the LOBSTER database.

Given that we do not observe separately the various components of the order flow as in \eqref{eq:lSPDEhom},   we use the relations discussed in Sec. \ref{ssec:obshape}
 to calibrate the parameters $\sigma$,  $\nu$ and the shape parameter 
\begin{equation}
  \label{eq:gamma}
  \gamma:= \frac{\beta}{2\eta}.
\end{equation}
for each side of the order book.
We set $L$ to the largest value  in our data set,
($L:=1000$). Parameters may be calibrated either through
\begin{enumerate}[label=(\alph*)]
\item a least squares fit of  \eqref{eq:average} to the average order book profile, or
 \item calibrating parameters to reproduce the position $\hat x$ and height of the maximum of the order book profile. 
\end{enumerate}
\begin{rmk}
  The estimator based on the
  maximum position of the peak is fast in computation but the fixed
  price level grid in the data restricts the set possible values for
  estimation of $\gamma$. In particular, the estimator is sensitive to
   the location of the maximum (i.e. the mode of the order book profile).
\end{rmk}
We show results for a set of NASDAQ stocks and
ETFs. Figure~\ref{fig:profile_fit_qqq} shows how the model reproduces the average book profile for QQQ at NASDAQ on 17th November 2017. In
Figure~\ref{fig:profile_fit_gamma} we see the coefficient $\gamma$
estimated across various 30-min windows during the trading day.
The
one-factor model based on the principal eigenfunction yields a
reasonable approximation for the average order book profile, which justifies our assumptions on the dynamics in~Section~\ref{ssec:setup_dyn}.

For low-price/large tick stocks, the average order book profiles may
differ from the exponential-sine shape. For such stocks, we  use
the nonlinear scaling described in Section~\ref{ssec:scaling}, leading to an average order book profile:
\begin{equation}
  \label{eq:9}  U(p) =  V\exp(-\gamma ((p-S_t)/\delta)^{a}) 
  \sin\left(((p-S_t)/\delta)^{a} \pi / L\right)),
\end{equation}
where $S_t $ is the best price. Figure~\ref{fig:profile_fit_siri} shows such a nonlinear fit for the average order book profile of
SIRI.

\begin{figure}
  \centering
  \includegraphics[width=\textwidth]{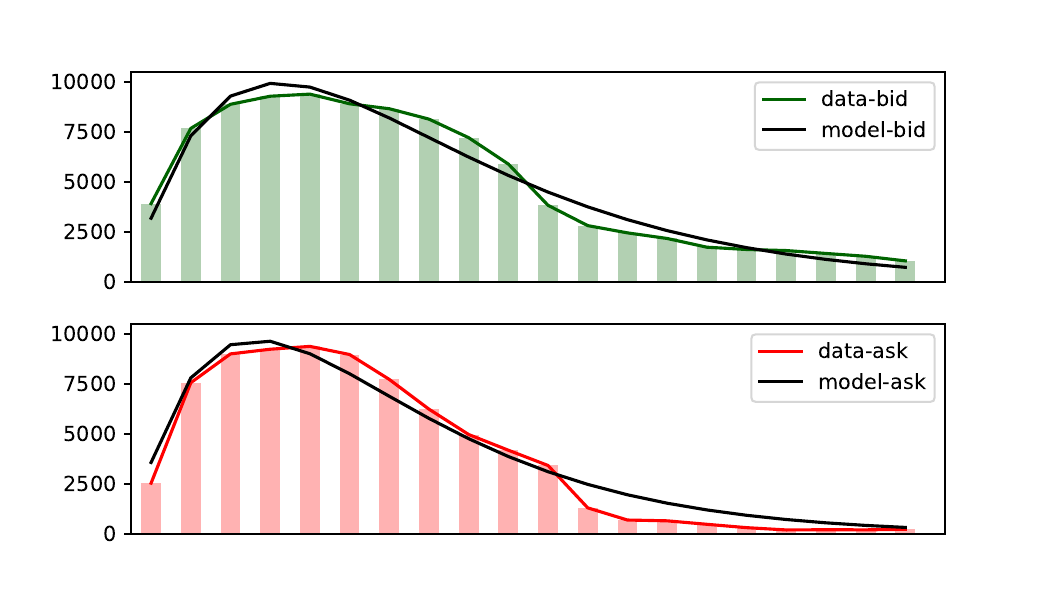}\\[-.85cm]
  \caption{Average profile of QQQ order book (first 20 levels), 17 Nov 2016 (Top: bid, Bottom: ask).}
  \label{fig:profile_fit_qqq}
\end{figure}
\begin{figure}
  \centering
  \includegraphics[width=\textwidth]{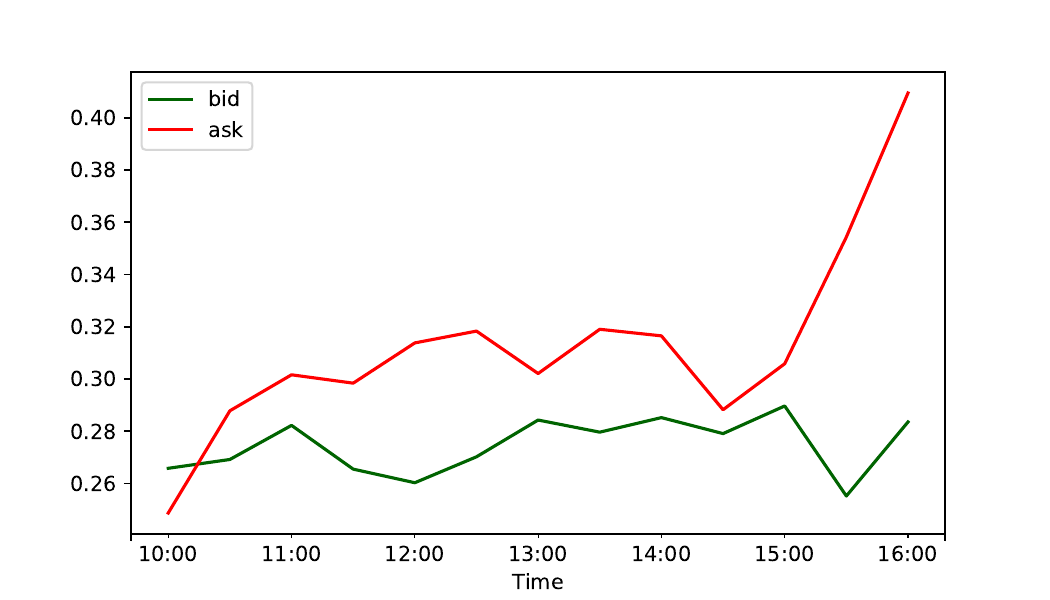}\\[-.35cm]
  \caption{Values of parameter $\gamma_a,\gamma_b$ estimated from 30 min average profile of  QQQ  order book (first 20 levels), 17 Nov 2016.}
  \label{fig:profile_fit_gamma}
\end{figure}
\begin{figure}
  \centering
  \includegraphics[width=\textwidth]{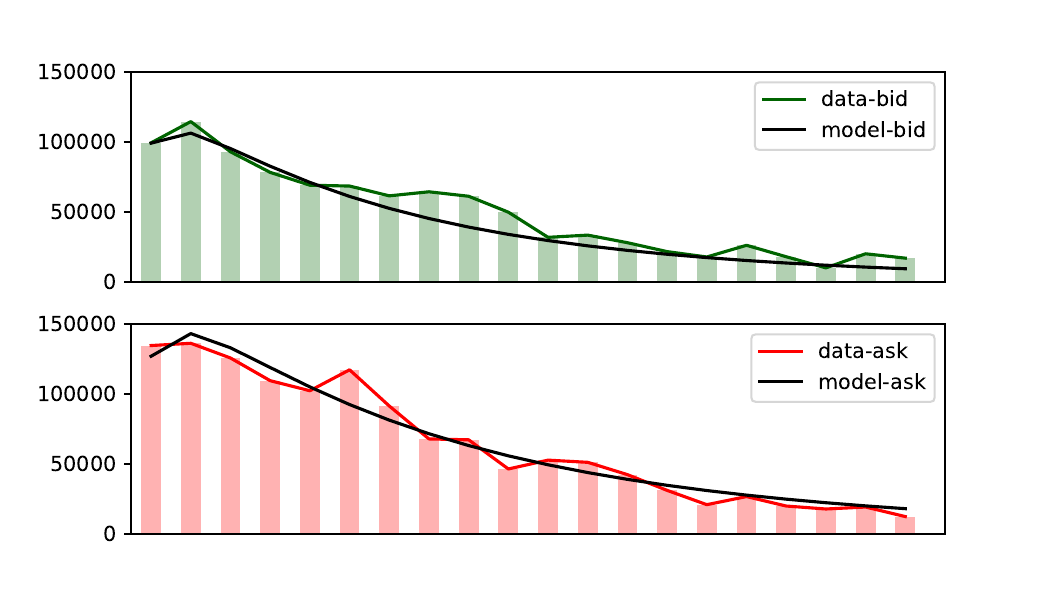}\\[-.85cm]
  \caption{Average profile of SIRI order book (first 20 levels, 17th November 2016)(Top: bid, Bottom: ask) $\gamma_b = 0.95$,
    $\gamma_a = 0.86$, $a_b= 0.52$, $a_a = 0.56$.}
  \label{fig:profile_fit_siri}
\end{figure}

\section{Mean-reverting models}
\label{sec:meanrev}
\subsection{A class of models with mean-reversion}

We now return to the full model \eqref{eq:lSPDE} with  non-zero source terms $ f^a(x), f^b(x)$ representing the rate of arrival of new limit orders at a distance $x$ from the best price:
\begin{equation}
 \nonumber
  \begin{split}
    \d u_t(x) &= \left[\eta_a \Delta u_t(x) + \beta_a \nabla u_t(x) + \alpha_a u_t(x) +  f^a(x)  \right]\d t+ \sigma_a u_t(x)  \d W^a_t , \quad \,x \in (0,L), \\
    \d u_t(x) &= \left[\eta_b \Delta u_t(x) - \beta_b \nabla u_t(x) + \alpha_b u_t(x) + f^b(x) \right]\d t  + \sigma_b u_t(x)  \d W^b_t, \quad \,x \in (-L, 0), \\
    u_t(0+) & = u_t(0-) = 0,\qquad
    u_t(-L)  = u_t(L) = 0,\quad  t>0,
  \end{split}
\end{equation}
with the sign condition
\begin{equation}\nonumber
  u_t(x) \leq 0, \quad x\in (-L,0),\quad \text{and}\quad u_t(x) \geq 0,\quad x\in (0,L), \;t\geq 0,  
\end{equation}
where, as above $\eta_a$, $\eta_b$, $\sigma_a$, $\sigma_b >0$,
$\beta_a$, $\beta_b \geq 0$, $\alpha_a$, $\alpha_b\in \R$ are constants and  $u_0\in L^2((-L,L))$. As above, we denote $u^b_0 :=
u_0\vert_{[-L,0]}$ and $u^a_0 := u_0\vert_{[0,L]}$.
 We will show that, when $\alpha_a$ and $\alpha_b$ are negative and $f^a(x) >0$, $f^b(-x)<0$ for all $x\in (0,L)$, this class of models leads to mean reverting
dynamics for the order book profile, consistent with the observation that
intraday dynamics of order book volume and queue size over
intermediate time scales (hours, day) typically exhibit mean reversion rather than a trend. 

Projecting the equation on the eigenfunctions $h_k^a,h_k^b$, as in Section \ref{sec:2factorLOB},  we see that, due to the fast increase in the eigenvalues \eqref{eq:spectrum}, solutions starting from a generic initial condition may be approximated by their projection on the principal eigenfunctions $h_1^a,h_1^b$ (we will justify this below in Proposition \ref{prop:asympLinInhom}) and the main contribution of heterogeneous order arrivals arises from the projection of  $f^a$ (resp. $f^b$) on $h_1^a$ (resp. $h_1^b$). 

 This motivates the following specfication, which leads to a tractable class of models:
\begin{equation}
  f^a(x) := \bar V_a\   H_1^a(x) ,\qquad f^b(x) := \bar V_b \ H_1^b(x), \quad \bar V_a>0, \quad \bar V_b >0 \label{eq:efinhom}
\end{equation}
Theorem~\ref{thm:inhomo} then gives explicit solutions to \eqref{eq:lSPDE}.
 Recall the notations~\eqref{eq:iprod_gamma+} and
\eqref{eq:iprod_gamma-} and define $V_t^b$ and $V_t^a$ by
\begin{equation}
  \label{eq:model_recgammadiff}
  \begin{alignedat}{2}
    \d V_t^a = \left(\bar V_a - \nu_a V_t^a\right) \d t + \sigma_a V_t^a \d W_t^a,&
    \d V_t^b = \left(\bar V_b - \nu_b V_t^b\right) \d t + \sigma_b V_t^b \d W_t^b
  \end{alignedat}  
\end{equation}
where $\nu_{i}:= \frac{\eta_{i} \pi^2}{L^2} +
\frac{\beta_{i}^2}{4\eta_{i}} -\alpha_{i}$, $i\in \{a,\,b\}$. The solution of the SPDE may then be obtained as follows:
\begin{prop} \label{prop:solutionLinInhom}
  \begin{enumerate}[label=(\roman*)]
  \item \label{i:inhom:general} The unique $L^2$-continuous solution
    of~\eqref{eq:lSPDE} -- \eqref{eq:efinhom} for a general initial condition $u_0$ is given by
    \begin{equation*}
      u_t(x) =
      \begin{cases}
        V_t^b H_1^b(x)  + \mathcal{E}_t(\sigma_b W^b) \sum_{k=1}^\infty e^{-\nu_k^bt}
        \iprod{u_0^b - V_0^b H_1^b}{h_k^b}{(-\frac{\beta_b}{2\eta_b})}
        h_k^b(x),&x\in (-L,0), \\
        V_t^a H_1^a(x)  + \mathcal{E}_t(\sigma_a W^a)\sum_{k=1}^\infty e^{-\nu_k^at}
        \iprod{u_0^a -V_0^a H_1^a}{h_k^a}{\frac{\beta_a}{2\eta_a}}
        h_k^a(x),&x\in 
        (0,L), \\
        0,\qquad x\notin 
        (-L,0)\cup (0,L).&
      \end{cases}          
    \end{equation*}
  \item \label{i:inhom:special} For an initial condition of the form
    \[ u_0(x) =        V_0^a H_1^a(x)
    \1_{[0,L]}+
        V_0^b H_1^b(x)\1_{[-L,0]} \]
    the unique $L^2$-continuous solution of~\eqref{eq:lSPDE} -- \eqref{eq:efinhom}  is given by
  \begin{equation}u_t(x) = \left(V_t^a H_1^a(x) \1_{(0,L)}(x) + V_t^b
        H_1^b(x) \1_{(-L,0)}(x)\right),\quad x\in [-L,L].\label{eq:proj} \end{equation}
  \end{enumerate}  
\end{prop}
\begin{proof}
  We obtain the general solution of the linear homogeneous equation
  from Proposition~\ref{prop:solutionLinHom}. The series  representation of $u$ 
  results from the spectral
decomposition, Proposition~\ref{prop:ev} and Theorem~\ref{thm:inhomo}.
\end{proof}

\subsection{Long time asymptotics and stationary solutions}

In order to derive properties of the 'average' order book profile, we now examine whether the order book profile $u_t$ has an ergodic behavior and describe stationary solutions. 
The following result  describes the long-term dynamics and shows that this dynamics is well approximated by projecting the initial condition on the principal eigenfunctions as done in \eqref{eq:proj}:
\begin{prop}\label{prop:asympLinInhom}
 Let $u_t$ be the unique solution  of
\eqref{eq:lSPDE} -- \eqref{eq:efinhom} for a general initial condition 
 $u_0\in L^2(-L,L)$ and define:
  \begin{equation}
    \label{eq:37}
    \check u_t (x):= V^b_t  H^b_1(x)  \1_{(-L,0)}(x) +  V^a_t  H^a_1(x)  \1_{(0,L)}(x),\qquad t>0.    
  \end{equation}  
      If $\nu_1^b>0$ and $\nu_1^a>0$, then:
  \begin{enumerate}[label=(\roman*)]
  \item\label{i:conv_unif}   The long-term dynamics of the order book is well approximated by the dynamics \eqref{eq:37} projected along the principal eigenfunctions:
    \begin{equation}
      \label{eq:46}
   \PP\left(   \lim_{t\to\infty} \norm{u_t -  \check u_t}{\infty} = 0\quad\right)=1.
    \end{equation}
    \item \label{i:conv_weak}$u_t$ has a unique stationary distribution and
    \begin{equation}
      \label{eq:36}
      u_t(x) \underset{t\to\infty}{\Longrightarrow} f^b(x) Z^b, \quad x<0,\qquad     u_t(x) \underset{t\to\infty}{\Longrightarrow} f^a(x) Z^a, \quad x>0,
    \end{equation}
    where $f^a,f^b$ are given by \eqref{eq:efinhom} and $Z^a$ (resp. $Z^b$ ) is an Inverse Gamma 
    random variable with shape parameter $1+2\frac{\nu_{a}}{\sigma_a^2}$ (resp. $1+2\frac{\nu_{b}}{\sigma_b^2}$) and scale
    parameter $\frac{\sigma_a^2}{2\bar V_{a}}$ (resp. $\frac{\sigma_b^2}{2\bar V_{b}}$). 
  \item\label{i:conv_L2} If  furthermore $\nu_1^b>\frac{\sigma_b^2}{2}$ and
    $\nu_1^a>\frac{\sigma^2_a}{2}$, then 
    \begin{equation}
      \label{eq:46}
      \lim_{t\to\infty} \E{\norm{u_t -  \check u_t}{L^2(-L,L)}^2}= 0.
    \end{equation}
  \end{enumerate} 
\end{prop}
\begin{proof}

  For $t_0>0$, let 
  \[K_{t_0} := \sqrt{\sum_{k=1}^\infty e^{-2(\nu_k^a - \nu_1^a)t_0}} <\infty.\]
  This term is indeed finite by integral criterion for series, see
  e.\,g. proof of Proposition~\ref{prop:solutionLinHom}. Denote  $u_t^\circ(.;h)$ the unique solution of the linear
  homogeneous equation~\eqref{eq:lSPDEhom} for an initial condition $h$. Recall from Theorem~\ref{thm:inhomo} that 
  \[ u_t(x) - \check u_t(x) = u_t^\circ(x; u_0-\check u_0).\]
  It suffices now to prove the results for the ask side and note that the
  calculations will be analogous for the bid side. 
  Using the representation of $u_t^\circ$ from
  Proposition~\ref{prop:solutionLinHom} we get for all $t> t_0$ and all $h\in L^2(0,L)$,
  \begin{align*}
    \norm{u_t^\circ(.; h)\vert_{(0,L)}}{\infty} 
    &\leq e^{-\nu_1^at}\mathcal{E}_t(\sigma_a W^a) \sqrt{\sum_{k=1}^\infty
      e^{-2(\nu_k^a-\nu_1^a) t_0}}\sqrt{\sum_{k=1}^\infty
      \abs{\scalp{h}{h^a_k}{\frac{\beta_a}{2\eta_a}}}^2} \\
    &= K_{t_0} \norm{h}{\frac{\beta_a}{2\eta_a}} \exp\left( \sigma_a
      W^a_t - \left(\nu_1^a+ \tfrac{\sigma_a^2}{2} \right)t\right),
  \end{align*}
  which, as $t\to \infty$, converges to $0$ provided that $\nu_1^a>0$. This proves \ref{i:conv_unif}. 

  To show \ref{i:conv_L2}, a similar calculation but using the
  orthogonality of the decomposition in
  Proposition~\ref{prop:solutionLinHom} yields
  \begin{align*}
    \E{\norm{u_t^\circ(.; h)\vert_{(0,L)}}{
\frac{\beta_a}{2\eta_a}}^{2}} 
    &= \sum_{k=1}^\infty e^{-2\nu_k^a t}\abs{\scalp{h}{h^a_k}{\frac{\beta_a}{2\eta_a}}}^2 \E{\abs{\mathcal{E}_t(\sigma_a W^a)}^2} \\
    &\leq e^{-2\nu_1^at}\norm{h}{\frac{\beta_a}{2\eta_a}}^2
      \E{\exp\left( 2\sigma_a W^a_t -\sigma_a^2 t\right)} \\
    &=e^{(-2\nu^a_1 +\sigma^2_a)t} \norm{h}{\frac{\beta_a}{2\eta_a}}^2.
  \end{align*}
  If $ \sigma_a^2< 2\nu_1^a$, then this converges to $0$ as $t\to
  \infty$. Since $\norm{.}{\frac{\beta_a}{2\eta_a}}$ defines an
  equivalent norm on $L^2(0,L)$, this finishes the proof of~\ref{i:conv_L2}.
  
 Assertion~\ref{i:conv_weak} follows from
  Proposition~\ref{prop:pearson_ergodicity}. Indeed, recall that $V^i, i\in \{a,b\}$ are  ergodic
  processes whose unique invariant distribution is given by an Inverse Gamma
  distribution with shape parameter $1+\frac{2\nu_{i}}{\sigma_i^2}$ and scale
  parameters $\frac{\sigma_i^2}{(\bar V_i)^2}$, $i\in
  \{a,\,b\}$. Denote by $Z^b$ and $Z^a$ random variables
  with these distribution For any $x\in [-L,L]$, we have the convergence in distribution
  \begin{equation}
    \check u_t\vert_{(-L,0)}\Longrightarrow  Z^b\ f^b_1(.),\qquad
    \check u_t \vert_{(0,L)} \Longrightarrow Z^a\ f^a(.).
    \label{eq:distr_ev}
  \end{equation}
  Since almost sure convergence yields convergence in distribution, by
  part~\ref{i:conv_unif} this yields that~\eqref{eq:distr_ev} holds also
  for $u_t$ with arbitrary initial data $u_0\in L^2(-L,L)$. 
\end{proof}

\subsection{Dynamics of order book volume}
Consider now the `projected' dynamics as in the setting of Proposition~\ref{prop:solutionLinInhom}.\ref{i:inhom:special}. 
The dynamics of the  order book volume $V_t$ is then given by
\begin{equation}
  \label{eq:11}
  V_t := \int_{-L}^L \abs{u_t(x)}\d x =  V^b_t + V_t^a,\qquad t\geq 0,
\end{equation}
where $V^b$ and $V^a$,  defined in \eqref{eq:model_recgammadiff}, represent
the volume of buy (resp. sell) orders in the order book.

Since $[W^a, W^b]_t = \varrho_{a,b} t$  we can write
\[ W^a =: W, \qquad W^b := \varrho_{a,b} W + \sqrt{1- \varrho_{a,b}^2} \widehat W,\]
for some Brownian motion $\widehat W$, independent of $W$. Then,
\begin{multline}
  \label{eq:40}
  \d V_t  = ( \bar V_a + \bar V_b - (\nu_a V_t^a + \nu_b V_t^b) \d
  t  \\
  + \left(\sigma_a V_t^a  + \varrho_{a,b}  \sigma_b V_t^b \right)\d W_t + \sqrt{1-\varrho_{a,b}^2} \sigma_b V_t^b \d \widehat W_t. 
\end{multline}
In particular, the quadratic variation (`realized variance') of the order book volume is given by
\begin{equation}
  \label{eq:realVar}
  \d \langle V \rangle_t = 
  \left(\sigma_a^2 (V_t^a)^2 +
    2\varrho_{a,b}\sigma_b \sigma_a V_t^a V_t^b+ \sigma_b^2 (V_t^b)^2
  \right) \d t
\end{equation}
For the symmetric and perfectly correlated case, $V$ is itself  
a reciprocal gamma diffusion:
\begin{cor}
  Assume the setting of
  Proposition~\ref{prop:solutionLinInhom}.\ref{i:inhom:special} and,
  in addition,
  that $\nu_a = \nu_b =:\nu$, $\sigma_a = \sigma_b =: \sigma$ and
  $\varrho_{a,b}=1$. Then, $V$ is the unique solution of
  \begin{equation}
    \label{eq:21}
    \d V_t = \left((\bar V_b + \bar V_a) - \nu
      V_t\right) \d t + \sigma V_t\d W_t,
  \end{equation}
  with $V_0 = V_0^b + V_0^a$. 
\end{cor}

In all cases, we get from~\eqref{eq:meanPearson} that for $i\in
\{a,b\}$, $t\geq 0$,
\begin{equation}
  \label{eq:42}
  \EE V_t^{i} = \left(V_0^{i} - \frac{\nu_i}{\bar V_i}\right) e^{-\nu_i t} +
  \frac{\nu_i}{\bar V_i}
\end{equation}
and
\begin{equation}
  \label{eq:42}
  \EE V_t^{i} = \left(V_0^{b} - \frac{\bar V_b}{\nu_b}\right) e^{-\nu_b t} +\left(V_0^{a} - \frac{\bar V_a}{\nu_a}\right) e^{-\nu_a t} +
  \frac{\bar V_b}{\nu_b} +
  \frac{\bar V_a}{\nu_a}.
\end{equation}

\subsection{Joint dynamics of mid-price and market depth}
\label{ssec:pricelininhom}
We now consider the mid price and market depths dynamics in the
situation of Proposition~\ref{prop:solutionLinInhom}.\ref{i:inhom:special}. As
discussed in Sections~\ref{sec:price} and
Section~\ref{ssec:pricelinh} for the linear homogeneous models,   the  dynamics of the mid-price is given by
\begin{equation*}
  \d S_t =    \theta \left(\frac{\d D_t^b}{D_t^b} - \frac{\d
      D_t^a}{D_t^a}\right),
\end{equation*}
where $\theta$ is an impact coefficient, while the bid/ask    depths follow
\begin{align*}
  D_t^a &:= \int_0^{\delta} u_t(x)\d x\approx \frac12\delta^2 \nabla
  u_t(0+)= \frac{\pi}{2L}\delta^2 V_t^a,\\
  D_t^b &:= -\int_{-\delta}^0 u_t(x) \d x \approx \frac12 \delta^2\nabla u_t(0-) =
  \frac{\pi}{2L}\delta^2 V_t^a.
\end{align*}
Thus, the dynamics of the market depths are given by
\begin{align*}
  \d D_t^b &= \nu_b\left(\overline{D}_b -  D_t^b\right) \d t + \sigma_b D_t^b \d W_t^b ,\\
  \d D_t^a &= \nu_a\left(\overline{D}_a -  D_t^a\right) \d t + \sigma_a D_t^a \d W_t^a .
\end{align*}
for some mean reversion levels $\overline{D}_b, \overline{D}_a>0$. We thus obtain the joint dynamics of 
price and market depth:
\begin{multline}
  \label{eq:price_depth_meanrev}
  \d
  \begin{pmatrix}
    D_t^b\\ D_t^a \\ S_t
  \end{pmatrix}
  =
  \begin{pmatrix}
    \nu_b(\overline{D}_b - D_t^b)\\
    \nu_a(\overline{D}_a - D_t^a)\\
     \theta\left( \frac{\nu_b\overline{D}_b}{D_t^b} - \frac{\nu_a\overline{D}_a}{D_t^a} -   (\nu_b - \nu_a)\right)
  \end{pmatrix}\d t\\
  + 
  \begin{pmatrix}
    \sigma_b D_t^b &0\\
    \varrho_{a,b}\sigma_a D_t^a & \sqrt{1-\varrho_{a,b}^2} \sigma_a
    D_t^a\\
      \theta\left(\sigma_b - \varrho_{a,b}\sigma_a\right) & -   \theta\sqrt{1-\varrho_{a,b}^2} \sigma_a
  \end{pmatrix}
  \d
  \begin{pmatrix}
    W^1_t\\ W_t^2
  \end{pmatrix},
\end{multline}
where $W^1$ and $W^2$ are independent Brownian motions.
The mid-price itself has quadratic variation $\langle S\rangle_t =
\sigma_S^2 t$, where
\begin{equation}
  \label{eq:midprice_vol}
  \sigma_S :=   \theta \sqrt{\sigma_b^2 + \sigma_a^2 - 2\sigma_a \sigma_b\varrho_{a,b}}.
\end{equation}
Over a small time interval $\Delta t,$
\begin{align*}
  S_{\Delta t} 
  &= S_0 +   \theta \int_0^{\Delta t}  \frac{\nu_b(\overline{D}_b -  D^b_s)}{D^b_s} - \frac{ \nu_a(\overline{D}_a - D^a_s)}{2D^a(s)} \d s+   \theta  \sigma_b  W_{\Delta t}^b -   \theta    \sigma_a W_{\Delta t}^a \\
  &\approx  S_0 + {\Delta t}   \frac{\theta}{2} \left(\frac{\nu_b (\overline{D}_b - D_{0}^b)}{D_{0}^b} - \frac{\nu_a(\overline{D}_a - D_{0}^a)}{ D_{0}^a}\right) + \sigma_S\sqrt{\Delta t}\ N_{0,1} 
\end{align*}
where $N_{0,1} $ is a standard Gaussian variable.
In particular the conditional probability of an upward mid-price move of size $y$ is given by
\begin{equation}\label{eq:approxpriceup}
  \P{ S_{\Delta t} \geq S_0 + y}
  \simeq N\left(
    \frac{  \theta \sqrt{{\Delta t}}}{\sigma_S}\left(\frac{\nu_b(\overline{D}_b - D_0^b)}{ D^b_{0}}
    - \frac{\nu_a(\overline{D}_a - D_0^a)}{D^a_{0}}\right) - \frac{y}{\sigma_S\sqrt{{\Delta t}}}\right),
\end{equation}
where $N$ denotes the cumulative distribution function of the standard normal distribution.
\begin{rmk}
  Using~\eqref{eq:meanPearson}, the expected order flow over a small time interval $[0,t]$ on each side of the book is given by
  for $\star\in \{a,b\}$,
  \begin{equation}
    \label{eq:30}
    \E{D_t^\star - D_0^\star} = t\nu_\star (\overline{D}_\star - D_0^\star) + o(t).
  \end{equation}
\end{rmk}
\begin{rmk}[Mean-reverting order book imbalance]
 The imbalance between buy and sell depth is a frequently used indicator for predicting short term price moves~\cite{cartea2018,cont2013,liptonImbalance}). In this model, the depth imbalance has the following dynamics:
  \[\d (D^b_t - D^a_t) = \left(\nu^b \overline{D}^b - \nu^a \overline{D}^a - (\nu^b D^b_t  - \nu^a D^a_t)\right) \d t + \sigma_bD^b_t \d W^b_t - \sigma_a D_t^a \d W_t^a.\]
  In the symmetric case, when $\overline{D}= \overline{D}_a = \overline{D}_b$, $\nu = \nu_a
  = \nu_b$, \eqref{eq:approxpriceup} becomes
  \begin{align}
    N\left(  \frac{\nu\overline{D}   \theta\sqrt{t}}{\sigma_S} \frac{\left(D^a_{0} - D^b_{0}\right)}{ D^a_{0} D^b_{0}} - \frac{y}{\sigma_S\sqrt{t}}\right).
  \end{align}
  This quantity is \emph{decreasing} in the depth imbalance $ D^b_{0} - D^a_{0}$:
  this is a consequence of the mean reversion in order book depth. In the symmetric case
  \begin{equation}
    \d (D^b_t - D^a_t) = -\nu \left(D^b_t  - D^a_t)\right) \d t + \sigma_bD^b_t \d W^b_t - \sigma_a D_t^a \d W_t^a,
  \end{equation}
  so the model reproduces the empirical observation that order book  imbalance is mean
  reverting \cite{cartea2018}.
  
Note that the model predicts mean reversion of market depths on the scale of
  $1/ \nu$ which corresponds to seconds for the ETFs QQQ and SPY and around 10 seconds for
 large tick stocks such as MSFT and INTC (see Table~\ref{tab:mix_av_30min}). For time scales smaller than  $1/ \nu$, the direction of price moves is highly correlated with order flow imbalance, as shown in  empirical studies  of equity markets  \cite{contImbalance}.
\end{rmk}

\subsection{Parameter estimation}
We now discuss estimation of model parameters from a discrete set of observations $( V^a_n, V^b_n)_{n=0,\ldots,N}$ of
the bid/ask volumes $V^a_t, V^b_t$ on a uniform time grid $\{k\Delta t \,\colon \,
k=0,\ldots ,N\}$. Let us rewrite the dynamics of
$V^a_t$ and $V^b_t$ in the form of reciprocal Gamma diffusions:
\begin{equation}
  \label{eq:19}
  \d V^\star_t = \nu_\star \left( \overline{D}_\star - V_0^\star
  \right) + \sqrt{2 \frac{ \nu_\star}{c_\star} (V_t^{\star})^2 } \d
  W^\star_t, \quad t\geq 0 ,\quad V_0^\star \in (0,\infty),\star\in \{a,b\}
\end{equation}
with $\nu_\star$, $\overline{D}_\star$, $c_\star>0$. We use method of moments estimators as
in~\cite{leonenko2010statistical}
for $\overline{D}_\star$ and $c_\star$ and a martingale estimation function \cite{bibby1995martingale} for
the autocorrelation parameters $\nu_\star$, $\star \in \{a,b\}$:
we define
\begin{equation}
  \nonumber
  \widehat{\overline{D}_\star} := \frac1{N} \sum_{k=1}^N \hat V_k,\quad{\rm and}\quad
  \hat{c_\star} := \frac{\sum_{n=1}^N (\hat V_n)^2}{\sum_{n=1}^N (\hat
    V_n)^2 - \widehat{ \overline{D}_\star^2 } }
    = 1 + 
    \frac{ \widehat{ \overline{D}_\star} ^2 }
    {\sum_{n=1}^N |\hat    V_n|^2 -  \widehat{\overline{D}_\star}^2 }.
\end{equation}
 Combining 
Proposition~\ref{prop:pearson_ergodicity} and
Remark~\ref{rmk:invgamma} with \cite[Theorem
6.3]{leonenko2010statistical} we obtain that if
$\overline{D}_\star>0$ and $c_\star > 5$, then $V^\star$ has finite $4$th moment and the estimators are
consistent and asymptotically normal.

For the autocorrelation parameters $\nu_a$ and $\nu_b$ we use the
martingale estimation function \cite[Section 2]{bibby1995martingale}:
\begin{equation}
  \label{eq:22}
  G_{\star}(\nu; \overline{D},c) := \frac{c}{\nu}\sum_{n=1}^N \frac{(\overline{D}_\star-\hat
    V^\star_{n-1})}{(\hat V_{n-1})^2}\left(\hat V_n - F(\hat V_{n-1}; \nu,\overline{D})\right),
\end{equation}
where
\begin{equation}
  \label{eq:24}
  F(z;\nu, \overline{D}) :=  (z - \overline{D}) e^{-\nu\Delta t} + \overline{D}.
\end{equation}
Given $\overline{D}_\star$, this yields the estimators
\begin{equation}
  \label{eq:26}
  \hat \nu_\star := \frac1{\Delta t} \log\left(-\frac{\sum_{n=1}^N \frac{(\overline{D}_\star -
        \hat V_{n-1})^2}{(\hat V_{n-1})^2} }{\sum_{n=1}^N \frac{(\overline{D}_\star -
        \hat V_{n-1})}{(\hat V_{n-1})^2} ( V_n - \overline{D}_\star)}\right),\qquad \star\in\{a,b\}.
\end{equation}
Convergence of this estimator is discussed in
 \cite[Theorem 3.2]{bibby1995martingale}.

We apply these estimators to high-frequency  limit order book time series for NASDAQ stocks and ETFs, obtained from the LOBSTER database, arranged into equally spaced observations over   time intervals of size $\Delta t= 10ms$ and $\d t = 50ms$. For each observation we use as market depth the average volume  of order in the first two price levels, respectively on bid and ask side.\footnote{The source code for
the implementation is  available online~\cite{lobpy}.}
Below we show  sample results for  ETFs (SPY and QQQ) and liquid stocks,(MSFT and INTC). 

Figure Table~\ref{tab:mix_av_30min} shows  estimated parameter values across different days for INTC, MSFT, QQQ and
SPY. We observe negative
values of correlation $\varrho_{a,b}$ across bid and ask order flows  which is consistent
with observations in~\cite{carmonaWebster}. 
\begin{table}[h]
  \centering
  \begin{tabular}{llccccccc}
  \toprule
  Ticker & Date &  $\mu_b$ &  $\mu_a $ &  $\nu_b$ &  $\nu_a $ &  $\sigma_b$ &  $\sigma_a $ &  $\varrho_{a,b}$ \\
  \midrule
  INTC   & 2016-11-15& 5179.0 &  5641.7 &   0.151 &   0.156 &      0.133 &      0.134 &       -0.077 \\
         & 2016-11-16& 5565.0 &  5672.5 &   0.082 &   0.118 &      0.111 &      0.124 &       -0.070 \\
         & 2016-11-17& 5776.5 &  7363.2 &   0.144 &   0.109 &      0.118 &      0.116 &       -0.019 \\
  \midrule
  MSFT   & 2016-11-15& 3035.6 &  3855.9 &   0.522 &   0.426 &      0.292 &      0.292 &       -0.092 \\
         & 2016-11-16& 2839.9 &  3562.1 &   0.409 &   0.395 &      0.239 &      0.240 &       -0.071 \\
         & 2016-11-17& 4149.0 &  5762.5 &   0.300 &   0.239 &      0.202 &      0.208 &       -0.146 \\
  \midrule  
  QQQ    & 2016-11-15& 4686.9 &  5489.2 &   2.467 &   1.972 &      0.724 &      0.639 &       -0.177 \\
         & 2016-11-16& 4801.0 &  5142.6 &   2.041 &   1.845 &      0.632 &      0.677 &       -0.177 \\
         & 2016-11-17& 6414.0 &  6226.4 &   1.428 &   1.281 &      0.510 &      0.506 &       -0.224 \\
    \midrule
  SPY    & 2016-11-15& 3903.4 &  4877.9 &   1.949 &   1.689 &      0.737 &      0.666 &       -0.176 \\
         & 2016-11-16& 3773.4 &  4486.4 &   1.324 &   1.763 &      0.578 &      0.657 &       -0.156 \\
         & 2016-11-17& 3693.0 &  4115.4 &   1.355 &   1.405 &      0.597 &      0.543 &       -0.181 \\
\bottomrule
\end{tabular}

  \caption{Averaged estimators for model parameters; $\nu$ and $\sigma$ are given per second.}
  \label{tab:mix_av_30min}
\end{table}
Figures~\ref{fig:cal_dyn} and~\ref{fig:cal_dyn2} show intraday variation of estimators for $\nu_a$, $\nu_b$, $\sigma_{a}$, $\sigma_b$
and $\varrho_{a,b}$ computed over 15-minute windows. 

There are various estimators for intraday price volatility in this model, which allows to test the model.
Recall that  in ~\eqref{eq:midprice_vol} we expressed price volatility  in terms of the parameters describing the order flow:
\begin{equation}
  \label{eq:pricevol_est}
  \hat \sigma_S := {\theta} \sqrt{\sigma_b^2 + \sigma_a^2 - 2 \sigma_b \sigma_a \varrho_{a,b}}.
\end{equation}
where $\theta$ is the impact coefficient.
We call this the RV estimator.

Another estimator is obtained by  first estimating $\sigma_b$ and $\sigma_a$
  using the martingale estimation function~\eqref{eq:22} then 
 computing the price volatility using  Equation~\eqref{eq:pricevol_est}. We label this the RCG estimator.
 
Finally, one can compute the realized variance of the price 
over a 30 minute time window using price changes over 10 ms intervals.
 Comparing these different estimators is a qualitative test of the model.

Figure~\ref{fig:pred_vola_etfs} compares these estimators, computed over 30 minute time windows: we observe that the model-based estimators are of the same order and closely track the intraday realized price volatility, which shows that the model captures correctly the qualitative relation between order flow and volatility.

\begin{figure}[p]
  \centering
  \vspace*{-.1cm}
  \includegraphics[width=0.8\textwidth, height=.52\textheight]{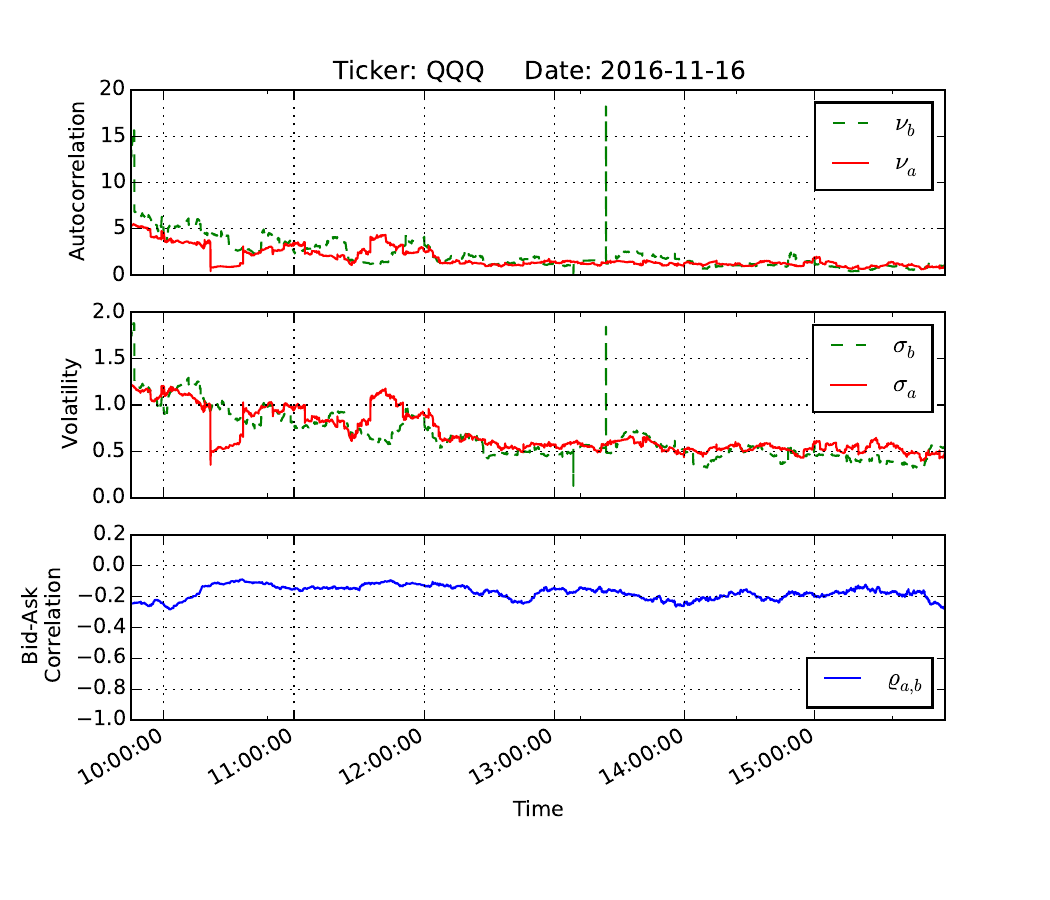}\\[-1.4cm]
  \includegraphics[width=\textwidth, height=.52
\textheight]{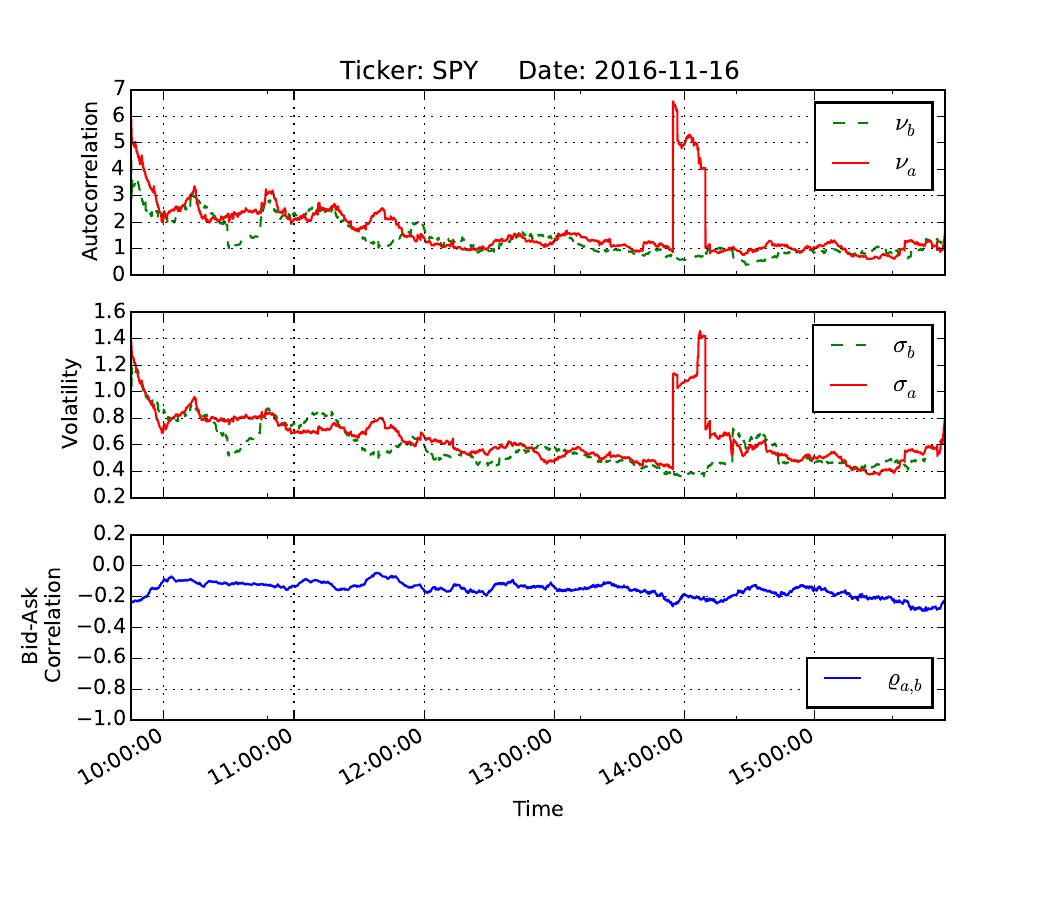}\\[-1.15cm]  
  \caption{Autocorrelation ($\nu_{a\slash b}$),
    standard deviation ($\sigma_{a\slash b}$) and bid-ask correlation
    ($\varrho_{a,b}$) of  order book depth in first 2 levels for two liquid ETFs (QQQ and SPY).}
  \label{fig:cal_dyn}
\end{figure}

\begin{figure}[p]
  \centering
  \vspace*{-.1cm}
  \includegraphics[width=\textwidth, height=.52\textheight]{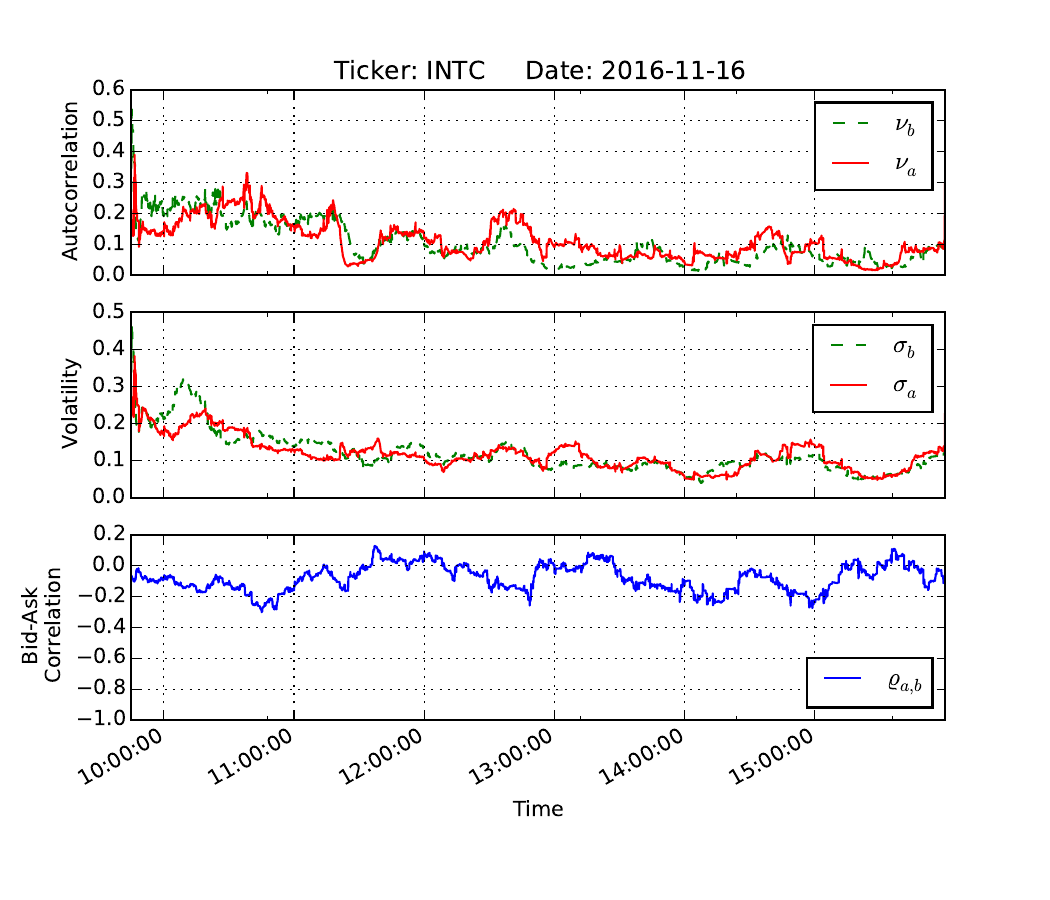}\\[-1.4cm]
  \includegraphics[width=\textwidth, height=.52\textheight]{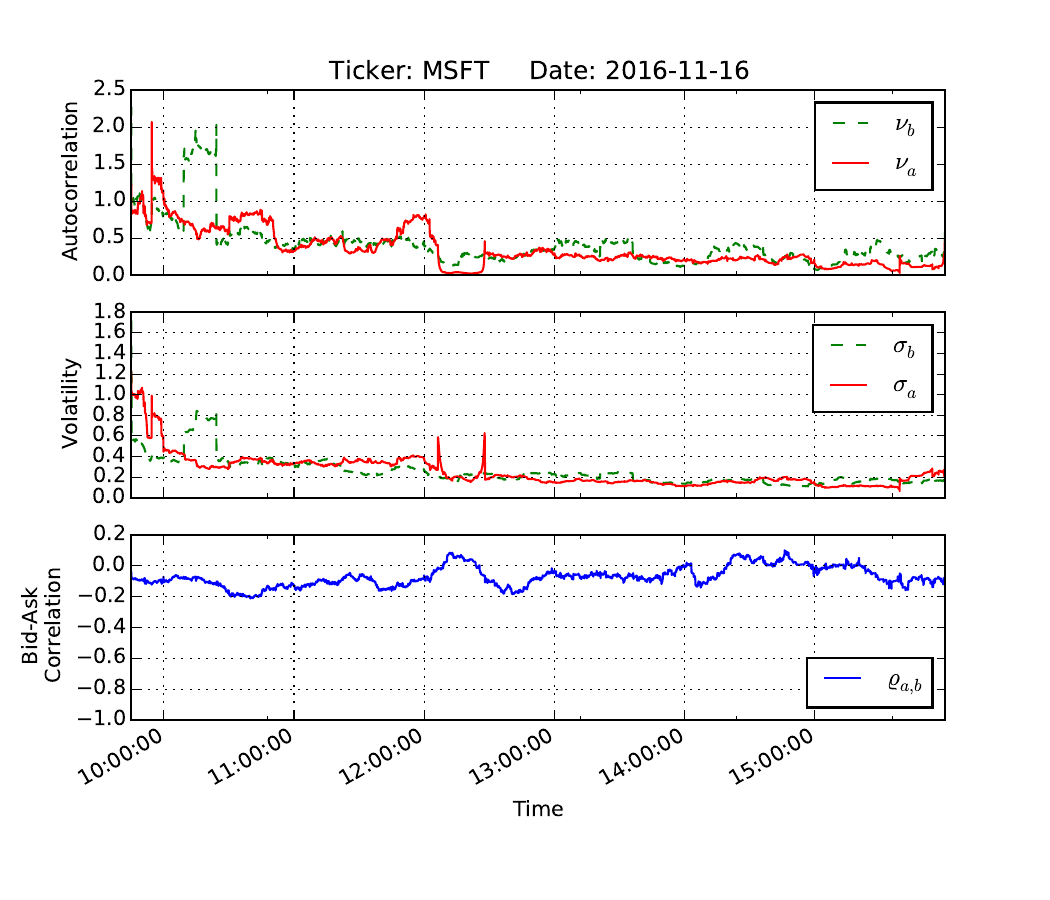}\\[-1.15cm]
  \caption{Autorcorrelation ($\nu_{a\slash b}$),
  standard deviation ($\sigma_{a\slash b}$) and bid-ask correlation
    ($\varrho_{a,b}$) of  order book depth in first 2 levels for two liquid stocks (INTC and MSFT).}
  \label{fig:cal_dyn2}
\end{figure}

\begin{figure}[p]
  \centering
  \vspace*{-2.cm}
  \includegraphics[width=\textwidth, height=1.2\textheight]{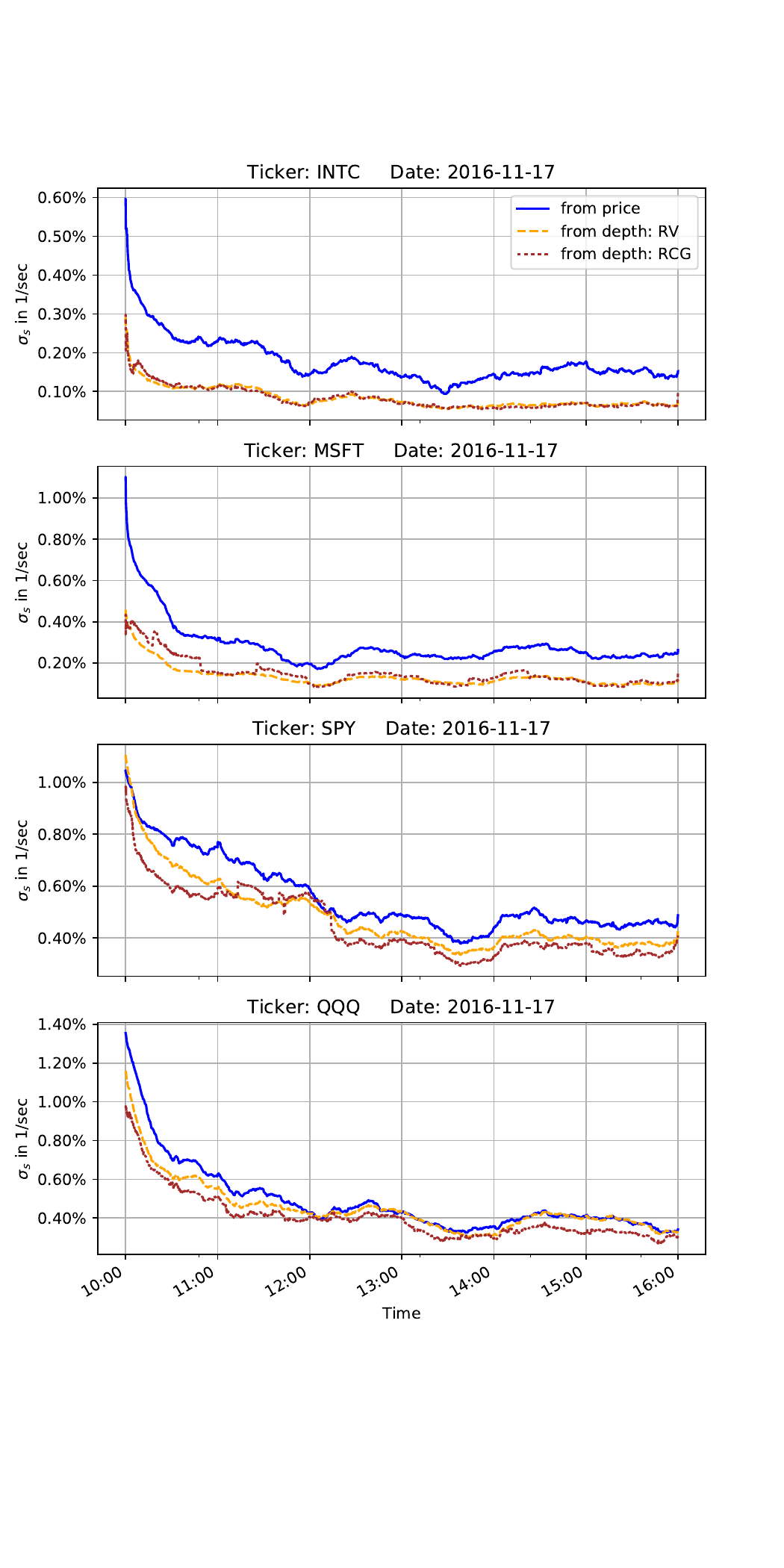}\\[-4.25cm]
  \caption{Comparison of various estimators for intraday price volatility $\sigma_S$:  
  standard deviation of price changes (blue), estimator based on realized variance/covariance of bid/ask depth (red), and estimator based on martingale estimation function (orange).}
  \label{fig:pred_vola_etfs}
\end{figure}

\clearpage

\newpage
\appendix
\section{Dynamics in absolute price coordinates}
\label{A:trafo}
We now discuss in more detail the generalized It\^{o}-Wentzell formula for distribution-valued processes, which is used in Section~\ref{ssec:2factor_noncentered} to derive the dynamics of the (non-centered) order book density
$v_t(p)$.
Let $C_0^\infty := C_0^\infty(\R)$ be
the space of smooth compactly supported functions on $\R$,  $\sD$ its dual, the space of generalized functions. We denote by $\ddx$ and $\ddxx$ the first two 
derivatives in the sense of distributions and by $\scalp{.}{.}{}$ the duality product on $\sD\times C_0^\infty$.

  A $\sD$-valued stochastic process $u=(u_t)_{t\geq 0}$ on  a filtered probability space $(\Omega, \mathcal{F},\mathbb{F},\mathbb{P})$ is called $\mathbb{F}$-\emph{predictable} if for all $\phi \in C^\infty_0(\R)$ the real valued process $\left( \ip{u_t}{\phi}{}\right)_{t\geq 0}$ is predictable. 

Let $N\in \N$ and $(b_t)_{t\geq 0}$ and $(c_t^k)_{t\geq 0}$, $k\in \{1,\ldots,N\}$ be predictable $\sD$ valued processes. We assume that for all $T$, $R\in (0,\infty)$ and all $\phi \in C^\infty_0(\R)$, almost surely
\begin{equation}
  \label{eq:int}
  \int_0^T \sup_{\abs{x}\leq R} \abs{\ip{b_t}{\phi(.-x)}{}} + \sum_{k=1}^N \abs{\ip{c_t^k}{\phi(.-x)}{}}^2 \d t < \infty.
\end{equation}

Let $(W_t^k, k=1,\ldots, N)_{t\geq 0}$ be  independent scalar Brownian motions. We  consider an equation of the form
\begin{equation}
  \label{eq:SPDEdistr}
  \d u_t = b_t \d t + \sum_{k=1}^N c_t^k \d W_t^k,
\end{equation}

\begin{df}\label{df:sol_distr}
  A $\sD$-valued stochastic process $(u_t)_{t\geq 0}$  is called a solution of~\eqref{eq:SPDEdistr} \emph{in the sense of distributions} with initial condition $u_0$ if for  $t\in (0,\infty)$ and  $\phi \in C^\infty_0$   \begin{equation}
    \label{eq:35}
    \ip{u_t}{\phi}{} - \ip{u_0}{\phi}{} = \int_0^t \ip{b_s}{\phi}{} \d t + \sum_{k=1}^N\int_0^t\ip{c_s^k}{\phi}{}{\d W_s^k}.
  \end{equation}
 holds  almost surely.
\end{df}

The following change of variable formula is a special case of a result by Krylov \cite[Theorem 1.1]{krylovItoWentzell}:
\begin{thm}[Generalized It\^{o}-Wentzell formula]\label{thm:itowentzell}
  Let $(u_t)_{t\geq 0}$ be a solution of ~\eqref{eq:SPDEdistr} in the sense of distributions
  and let $(x_t)_{t\geq 0}$ be a locally
  integrable  process with representation
  \begin{equation*}
    \d x_t = \mu_t \d t +\sum_{k=1}^N \sigma^k_t \d W_t^k,\qquad t\geq 0.
  \end{equation*}
  where $(\mu_t)_{t\geq 0}$ and $(\sigma^k_t, k=1..N)_{t\geq 0}$  are real-valued predictable  processes.
  Define the $\sD$-valued process $(v_t)_{t\geq 0}$ by $v_t(x) := u_t(x+x_t)$, for $x\in \R$, $t\in [0,\infty)$. Then $(v_t)_{t\geq 0}$ is a solution of  
  \begin{multline*}
    \d v_t = \left[ b_t(.+x_t) + \tfrac12 \left(\sum_{k=1}^N \abs{\sigma_t^k}^2\right) \ddxx v_t + \mu_t \ddx v_t + \sum_{k=1}^N \left(\sigma_t^k \ddx c_t^k(.+ x_t)\right)\right] \d t\\
    +\sum_{k=1}^N\left[ c_t^k(.+x_t) + \sigma_t^k \ddx v_t \right] \d W_t^k
  \end{multline*}
  in the sense of distributions.
\end{thm}
\begin{rmk}
  It is worth noting that the correlation of $(u_t)$ and $(x_t)$ contributes  the term
  \[ \sum_{k=1}^N \left(\sigma_t^k \ddx c_t^k(.+ x_t)\right).\]
\end{rmk}

We now apply the above It\^{o}-Wentzell formula in order to derive the dynamics of the order book density $v$, in non-centered coordinates, in the setting considered in Sections~\ref{sec:2factorLOB} and~\ref{sec:meanrev}.

Let $L\in (0,\infty]$ and $I:= (-L,0)\cup (0,L)$. For $h$, $f\in
H^2(I)\cap H^1_0(I)$. Then,~\eqref{eq:lSPDE} with initial
condition $u_0 = h$ admits a unique (analytically) strong solution
denoted by $(u_t)_{t\geq 0}$ . Let $\tilde u_t$ be the trivial extension of $u_t$ to $\R$, i.\,e.
\begin{equation}
  \label{eq:20}
  \tilde u_t(x) :=
  \begin{cases}
    u_t(x),&x\in I,\\
    0,& \text{ otherwise}.
  \end{cases}
\end{equation}
Note that $\tilde u \in H^2(\R\setminus\{-L, 0, L\})\cap
H^1(\R)$. Recall that $\Delta$ and $\nabla$ in the previous discussions denoted the weak derivatives on $\R\setminus\{-L,0,L\}$, and we get that
$\ddx \tilde u  =  \nabla \tilde u $
and
\begin{equation}
  \label{eq:33}
  \begin{aligned}
    \ddxx \tilde u - \Delta \tilde u
    &= \ddx \nabla \tilde u - \nabla \nabla \tilde u\\
    &= (\nabla \tilde u(-L+) - \nabla \tilde u(-L-)) \delta_{-L}\\
    &+(\nabla \tilde u(0+) - \nabla \tilde u(0-)) \delta_0 +(\nabla \tilde u(L+) - \nabla \tilde u(L-)) \delta_{L},
  \end{aligned}
\end{equation}
where $\delta_x$ denotes a point mass at $x\in \R$. Define
\begin{align}
  \label{eq:34}
  b_t(x) &:=
           \begin{cases}
             \eta_a \Delta u_t(x) + \beta_a \nabla u_t(x) + \alpha_a
             u_t(x) + f_a(x),&  x\in (0,L),\\
             \eta_b \Delta u_t(x) - \beta_b \nabla u_t(x) + \alpha_b
             u_t(x) - f_b(x),& x\in (-L,0),\\
             0,&\text{ otherwise},
           \end{cases}\\
  c_t^1(x) &:=
             \begin{cases}
               \sigma_a\varrho_{a,b} u_t(x),& x\in (0,L),\\
               \sigma_b u_t(x),& x\in (-L,0),\\
               0,& \text{ otherwise},
             \end{cases}\\
  c_t^2(x) &:=
             \begin{cases}
               \sigma_a\sqrt{1-\varrho_{a,b}^2} u_t(x),& x\in (0,L),\\
               0,& \text{ otherwise},
             \end{cases}    
\end{align}
so that
\begin{equation}
\label{eq:41}
\d \tilde u_t = b_t \d t + c_t^1 \d W_t^1 + c_t^2 \d W_t^2.
\end{equation}
The Cauchy-Schwartz inequality shows  that \eqref{eq:int} is satisfied. Assume now that the mid price $(S_t)_{t\geq 0}$ follows the dynamics 
\begin{equation}
\label{eq:A:midprice}
\d S_t = c_s\theta \mu_t \d t +   c_s\theta  (\sigma_b - \sigma_a \varrho_{a,b}) \d
W_t^1 -  c_s\theta  \sigma_a \sqrt{1-\varrho_{a,b}^2} \d W_t^2. 
\end{equation}
for some  integrable predictable process $\mu$. Define
\begin{equation}
\label{eq:32}
\sigma_s:=  c_s\theta  \sqrt{\sigma_b^2 + \sigma_a^2 - 2 \varrho_{a,b} \sigma_b\sigma_a}.
\end{equation}
Then, Theorem~\ref{thm:itowentzell} yields that for $ v_t(x) := \tilde u_t(x-S_t)$ we get
\begin{equation}
  \begin{aligned}
    \d v_t &= \left[b_t(.-S_t) + \tfrac12 \sigma_s^2 \ddxx v_t-c_s\theta \mu_t \ddx v_t \right.\\
    &\qquad\left. -\left( c_s\theta\left(\sigma_b - \varrho_{a,b} \sigma_a\right) \ddx c_t^1(.-S_t) + c_s\theta \sqrt{1-\varrho_{a,b}^2} \sigma_a \ddx c_t^2(.-S_t)\right) \right]\d t\\
    &\qquad + \left( c_t^1(.-S_t) - c_s\theta\left(\sigma_b - \varrho_{a,b} \sigma_a\right) \ddx v_t\right) \d W_t^1\\
    &\qquad + \left( c_t^2(.-S_t) + c_s\theta \sqrt{1-\varrho_{a,b}^2} \sigma_a \ddx v_t\right)\d W_t^2
  \end{aligned}
\end{equation}
i.\,e. $v$ is a solution of the stochastic moving boundary problem, 
\begin{equation}
\begin{alignedat}{2}
  \d v_t &= \left[\left(\eta_a +\tfrac12 \sigma_s^2\right) \Delta v_t + \left(\beta_a-c_s\theta\mu_t - c_s\theta\left( \varrho_{a,b}\sigma_b\sigma_a - \sigma_a^2\right)\right) \nabla v_t  \right. &&\\
  &\qquad\hspace{12em} \left. \vphantom{\left(\eta_a +\tfrac12 \sigma_s^2\right) \Delta v_t + \left(\beta_a-c_s\theta\mu_t - c_s\theta\left( \varrho_{a,b}\sigma_b\sigma_a - \sigma_a^2\right)\right) \nabla v_t }
    + \alpha_a  v_t + f_a(.-S_t) \right] \d t &&\\
  &\quad + \left(\sigma_a \varrho_{a,b}v_t  - c_s\theta\left(\sigma_b - \varrho_{a,b} \sigma_a\right) \nabla v_t \right) \d W_t^1 &&\\
  &\quad + \sigma_a \sqrt{1-\varrho_{a,b}^2}\left(  v_t + c_s\theta \nabla v_t\right)\d W_t^2, \ \text{ on }\ (S_t, S_t + L), & & \\
  \d v_t &= \left[\left(\eta_b +\tfrac12 \sigma_s^2\right) \Delta v_t  - \left(\beta_b + c_s\theta\mu_t + c_s\theta\left(\sigma_b^2 - \varrho_{a,b} \sigma_b\sigma_a\right) \right) \nabla v_t \right.&&\\
  &\qquad\hspace{12em}
  \left.\vphantom{\left(\eta_b +\tfrac12 \sigma_s^2\right) \Delta v_t
      \left(\beta_b + c_s\theta \mu_t + c_s\theta\left(\sigma_b^2 - \varrho_{a,b} \sigma_b\sigma_a\right) \right) \nabla v_t }
    + \alpha_b  v_t - f_b(.-S_t) \right] \d t &&\\
  &\quad + \left(\sigma_b v_t  - c_s\theta\left(\sigma_b - \varrho_{a,b} \sigma_a\right) \nabla v_t \right) \d W_t^1 &&\\
  &\quad + c_s\theta \sqrt{1-\varrho_{a,b}^2} \sigma_a \nabla v_t\d W_t^2  \ \text{ on }\ (S_t-L, S_t),&&  \\
  v_t  &= 0 \qquad  \text{ otherwise;} && \\
  \d S_t &= c_s\theta\mu_t \d t + c_s\theta (\sigma_b  - \varrho_{a,b} \sigma_a) \d W_t^1 - c_s\theta \sqrt{1-\varrho_{a,b}^2} \sigma_a \d W_t^2,&&
\end{alignedat} \label{eq:smbp}
\end{equation}
To define what we mean by solution in this context we introduce the
mappings 
\begin{align*}
  \label{eq:43}
  \hspace{10em}&\hspace{-10em}
                 \mathbb{L} \colon \bigcup_{x\in \R} \left[\left(H^{2}(\R\setminus\{x-L,x,x+L\}) \cap H^1_0(\R\setminus\{x-L,x,x+L\})\right)  \times \{x\}\right] \to \sD,\\ \quad
  (v,s)&\mapsto (\nabla (v(s-L+) - \nabla v(s-L-)))\delta_{s-L} \\
               &\qquad+(\nabla (v(s+) - \nabla v(s-)))\delta_{s} \\
               &\qquad+(\nabla (v(s+L+) - \nabla v(s+L-)))\delta_{s+L}.
\end{align*}
 Define now the functions $\bar \mu \colon \R^5
\to \R$, $\bar \sigma_1$, $\bar \sigma_2 \colon \R^4 \to \R$  as
\begin{align*}
  \bar \mu(x,y'', y', y, s) &:=
                              \begin{cases}
                                (\eta_a + \tfrac12 \sigma_s^2) y'' + (\beta_a - c_s\theta (\varrho_{a,b} \sigma_b \sigma_a - \sigma_a^2)) y' &\\
                                \hfill + \alpha_a y + f_a(x),& x\in (0,L)\\
                                (\eta_b + \tfrac12 \sigma_s^2) y'' - (\beta_a + c_s\theta (\sigma_b^2 - \varrho_{a,b} \sigma_b \sigma_a )) y'&\\
                                \hfill+ \alpha_b y - f_b(x),& x\in (-L,0),\\
                                0,&\text{ otherwise},
                              \end{cases}
  \\
  \bar \sigma_1(x,y',y,s) &:=
                            \begin{cases}
                              \sigma_a \varrho_{a,b} y  ,& x\in (0,L),\\
                              \sigma_b y , &x\in (-L,0),\\
                              0,&\text{ otherwise},
                            \end{cases}\\
  \bar \sigma_2(x,y',y,s) &:=
                            \begin{cases}
                              \sigma_a \sqrt{1-\varrho_{a,b}^2}y, & x\in (0,L)\\
                              0, \quad\text{ otherwise} &
                            \end{cases}
\end{align*}
for $x$, $y''$, $y'$, $y$, $s\in \R$.

Following \cite[Definition 1.11]{mueller2016stochastic}, a solution
of~\eqref{eq:smbp} is an $L^2(\R)\times \R$-continuous stochastic process $(v_t,S_t)$, taking values in
\[ \bigcup_{x\in \R} \left[\left(H^{2}(\R\setminus\{x-L,x,x+L\}) \cap H^1_0(\R\setminus\{x-L,x,x+L\})\right)  \times \{x\}\right],\]
such that $(S_t)$ is given by~\eqref{eq:A:midprice} and, in the sense of distributions,
\begin{multline}
  \d v_t = \left(\bar \mu(.-S_t, \Delta v_t, \nabla v_t, v_t, S_t)  \right) \d t - \nabla v_t \d S_t + \tfrac12 \mathbb{L}(v_t,S_t)\d \langle S\rangle_t \\
  + \bar \sigma_1 (.-S_t, \nabla v_t, v_t, S_t) \d W_t^1 + \bar \sigma_2(.-S_t, \nabla v_t, v_t ,S_t)\d W_t^2.
\end{multline}

\bibliographystyle{apalike}
\bibliography{literature}

\begin{thebibliography}{}

\bibitem[Bibby et~al., 2005]{bibby2005diffusion}
Bibby, B.~M., Skovgaard, I.~M., and S{\o}rensen, M. (2005).
\newblock Diffusion-type models with given marginal distribution and
  autocorrelation function.
\newblock {\em Bernoulli}, 11(2):191--220.

\bibitem[Bibby and S{\o}rensen, 1995]{bibby1995martingale}
Bibby, B.~M. and S{\o}rensen, M. (1995).
\newblock Martingale estimation functions for discretely observed diffusion
  processes.
\newblock {\em Bernoulli}, 1(1-2):17--39.

\bibitem[Borodin and Salminen, 2012]{borodin2012handbook}
Borodin, A.~N. and Salminen, P. (2012).
\newblock {\em Handbook of Brownian motion-facts and formulae}.
\newblock Birkh{\"a}user.

\bibitem[Bouchaud et~al., 2009]{bouchaud2009}
Bouchaud, J.-P., Farmer, J., and Lillo, F. (2009).
\newblock How markets slowly digest changes in supply and demand.
\newblock In Hens, T. and Schenk-Hoppe, K.~R., editors, {\em { Handbook of
  financial markets: dynamics and evolution}}, pages 57--160. Elsevier.

\bibitem[Burger et~al., 2013]{burger2013}
Burger, M., Caffarelli, L., Markowich, P.~A., and Wolfram, M.-T. (2013).
\newblock On a {B}oltzmann-type price formation model.
\newblock {\em {Proceedings of the Royal Society of London A: Mathematical,
  Physical and Engineering Sciences}}, 469(2157).

\bibitem[Caffarelli et~al., 2011]{caffarelli2011}
Caffarelli, L.~A., Markowich, P.~A., and Pietschmann, J.-F. (2011).
\newblock On a price formation free boundary model by {Lasry and Lions}.
\newblock {\em Comptes Rendus Mathematique}, 349(11):621 -- 624.

\bibitem[{Carmona} and {Webster}, 2013]{carmonaWebster}
{Carmona}, R. and {Webster}, K. (2013).
\newblock {The Self-Financing Equation in High Frequency Markets}.
\newblock ArXiv 1312.2302.

\bibitem[Cartea et~al., 2018]{cartea2018}
Cartea, A., Donnelly, R., and Jaimungal, S. (2018).
\newblock Enhancing trading strategies with order book signals.
\newblock {\em Applied Mathematical Finance}, 25(1):1--35.

\bibitem[Chavez-Casillas and Figueroa-Lopez, 2017]{Chavez2017}
Chavez-Casillas, J.~A. and Figueroa-Lopez, J. (2017).
\newblock A one-level limit order book model with memory and variable spread.
\newblock {\em Stochastic Processes and their Applications}, 127(8):2447 --
  2481.

\bibitem[Chayes et~al., 2009]{chayes2009}
Chayes, L., del Mar~Gonzalez, M., Gualdani, M., and Kim, I. (2009).
\newblock Global existence and uniqueness of solutions to a model of price
  formation.
\newblock {\em {SIAM Journal on Mathematical Analysis}}, 41(5):2107--2135.

\bibitem[Cont, 2005]{cont2005}
Cont, R. (2005).
\newblock Modeling term structure dynamics: an infinite dimensional approach.
\newblock {\em Int. Journal of Theoretical and Applied finance},
  8(03):357--380.

\bibitem[Cont, 2011]{IEEE}
Cont, R. (2011).
\newblock Statistical modeling of high-frequency financial data.
\newblock {\em IEEE Signal Processing}, 28(5):16--25.

\bibitem[Cont and De~Larrard, 2012]{contlarrard2012}
Cont, R. and De~Larrard, A. (2012).
\newblock Order book dynamics in liquid markets: Limit theorems and diffusion
  approximations.

\bibitem[Cont and de~Larrard, 2013]{cont2013}
Cont, R. and de~Larrard, A. (2013).
\newblock Price dynamics in a {M}arkovian limit order market.
\newblock {\em SIAM J. Financial Math.}, 4(1):1--25.

\bibitem[Cont et~al., 2014]{contImbalance}
Cont, R., Kukanov, A., and Stoikov, S. (2014).
\newblock The price impact of order book events.
\newblock {\em Journal of Financial Econometrics}, 12(1):47--88.

\bibitem[Cont and Mueller, 2018]{lobpy}
Cont, R. and Mueller, M. (2018).
\newblock {\em {LOBPY} Python package: {\tt https://github.com/mm842/lobpy}}.

\bibitem[Cont et~al., 2010]{cont2010}
Cont, R., Stoikov, S., and Talreja, R. (2010).
\newblock A stochastic model for order book dynamics.
\newblock {\em Oper. Res.}, 58(3):549--563.

\bibitem[Da~Prato and Zabczyk, 2014]{dPZinf}
Da~Prato, G. and Zabczyk, J. (2014).
\newblock {\em Stochastic equations in infinite dimensions}.
\newblock Cambridge University Press, 2nd edition.

\bibitem[Evans, 2010]{Evans2010}
Evans, L.~C. (2010).
\newblock {\em Partial Differential Equations}, volume~19.
\newblock American Mathematical Society, 2nd edition.

\bibitem[Filipovic and Teichmann, 2003]{filipovic2003}
Filipovic, D. and Teichmann, J. (2003).
\newblock Existence of invariant manifolds for stochastic equations in infinite
  dimension.
\newblock {\em Journal of Functional Analysis}, 197(2):398 -- 432.

\bibitem[Forman and S{\o}rensen, 2008]{forman2008pearson}
Forman, J.~L. and S{\o}rensen, M. (2008).
\newblock {P}earson diffusions: A class of statistically tractable diffusion
  processes.
\newblock {\em Scandinavian Journal of Statistics}, 35(3):438--465.

\bibitem[Gaspar, 2006]{gaspar2006}
Gaspar, R. (2006).
\newblock Finite dimensional {M}arkovian realizations for forward price term
  structure models.
\newblock In Shiryaev, A., Grossinho, M., Oliveira, P., and Esquivel, M.,
  editors, {\em { Stochastic Finance}}, pages 265--320. Springer.

\bibitem[Hambly et~al., 2020]{hambly2018}
Hambly, B., Kalsi, J., and Newbury, J. (2020).
\newblock Limit order books, diffusion approximations and reflected {SPDE}s:
  From microscopic to macroscopic models.
\newblock {\em Applied Mathematical Finance}, 27(1-2):132--170.

\bibitem[Horst and Kreher, 2018]{horst2018}
Horst, U. and Kreher, D. (2018).
\newblock Second order approximations for limit order books.
\newblock {\em Finance and Stochastics}, 22(4):827--877.

\bibitem[Huang et~al., 2017]{rosenbaum2017}
Huang, W., Lehalle, C.-A., and Rosenbaum, M. (2017).
\newblock {Simulating and analyzing order book data: The queue-reactive model}.
\newblock {\em Journal of the American Statistical Association}, 110:107--122.

\bibitem[Kallenberg, 2002]{kallenberg}
Kallenberg, O. (2002).
\newblock {\em Foundations of Modern Probability}.
\newblock Springer.

\bibitem[Karatzas and Kardaras, 2007]{karatzas2007numeraire}
Karatzas, I. and Kardaras, C. (2007).
\newblock The num{\'e}raire portfolio in semimartingale financial models.
\newblock {\em Finance and Stochastics}, 11(4):447--493.

\bibitem[Karatzas and Shreve, 1987]{karatzas2012brownian}
Karatzas, I. and Shreve, S. (1987).
\newblock {\em Brownian motion and stochastic calculus}.
\newblock Springer.

\bibitem[Kelly and Yudovina, 2018]{kelly2018}
Kelly, F. and Yudovina, E. (2018).
\newblock A {Markov} model of a limit order book: Thresholds, recurrence, and
  trading strategies.
\newblock {\em Mathematics of Operations Research}, 43(1):181--203.

\bibitem[Kirilenko et~al., 2017]{kirilenko2017}
Kirilenko, A., Kyle, A., Samadi, M., and Tuzun, T. (2017).
\newblock The flash crash: High frequency trading on an electronic market.
\newblock {\em {Journal of Finance}}, {LXXII}:967--998.

\bibitem[Kloeden and Platen, 1992]{Kloeden1992}
Kloeden, P.~E. and Platen, E. (1992).
\newblock {\em Numerical Solution of Stochastic Differential Equations}.
\newblock Springer.

\bibitem[Krylov, 2011]{krylovItoWentzell}
Krylov, N.~V. (2011).
\newblock On the {I}t\^o-{W}entzell formula for distribution-valued processes
  and related topics.
\newblock {\em Probab. Theory Related Fields}, 150(1-2):295--319.

\bibitem[Lasry and Lions, 2007]{lasrylions}
Lasry, J.-M. and Lions, P.-L. (2007).
\newblock Mean field games.
\newblock {\em Jpn. J. Math.}, 2(1):229--260.

\bibitem[Lehalle and Laruelle, 2018]{lehalle2018}
Lehalle, C.-A. and Laruelle, S. (2018).
\newblock {\em Market microstructure in practice}.
\newblock World {S}cientific.

\bibitem[Leonenko and {\v{S}}uvak, 2010]{leonenko2010statistical}
Leonenko, N. and {\v{S}}uvak, N. (2010).
\newblock Statistical inference for reciprocal gamma diffusion process.
\newblock {\em Journal of statistical planning and inference}, 140(1):30--51.

\bibitem[L\'evine, 1991]{levine1991}
L\'evine, J. (1991).
\newblock Finite dimensional realizations of stochastic {PDEs} and application
  to filtering.
\newblock {\em Stochastics and Stochastic Reports}, 37(1-2):75--103.

\bibitem[{Lipton} et~al., 2014]{liptonImbalance}
{Lipton}, A., {Pesavento}, U., and {Sotiropoulos}, M.~G. (2014).
\newblock Trading strategies via book imbalance.
\newblock {\em RISK}, pages 70--75.

\bibitem[Luckock, 2003]{luckock}
Luckock, H. (2003).
\newblock A steady-state model of the continuous double auction.
\newblock {\em Quantitative Finance}, 3(5):385--404.

\bibitem[Markowich et~al., 2016]{markowich2016}
Markowich, P., Teichmann, J., and Wolfram, M. (2016).
\newblock Parabolic free boundary price formation models under market size
  fluctuations.
\newblock {\em Multiscale Modeling \& Simulation}, 14(4):1211--1237.

\bibitem[Milian, 2002]{milianComp}
Milian, A. (2002).
\newblock Comparison theorems for stochastic evolution equations.
\newblock {\em Stoch. Stoch. Rep.}, 72(1-2):79--108.

\bibitem[Mueller, 2018]{mueller2016stochastic}
Mueller, M.~S. (2018).
\newblock {A stochastic Stefan-type problem under first-order boundary
  conditions}.
\newblock {\em The Annals of Applied Probability}, 28(4):2335--2369.

\bibitem[Revuz and Yor, 1999]{ry}
Revuz, D. and Yor, M. (1999).
\newblock {\em {Continuous Martingales and Brownian Motion}}.
\newblock Springer.

\bibitem[Schlichting and Gersten, 2017]{boundary}
Schlichting, H. and Gersten, K. (2017).
\newblock {\em Boundary-layer theory}.
\newblock Graduate Texts in Mathematics. Springer-Verlag.

\bibitem[Smith et~al., 2003]{smith2003}
Smith, E., Farmer, J.~D., Gillemot, L., and Krishnamurthy, S. (2003).
\newblock Statistical theory of the continuous double auction.
\newblock {\em Quantitative Finance}, 3(6):481--514.

\bibitem[Yosida, 1995]{yosida1995functional}
Yosida, K. (1995).
\newblock {\em Functional analysis}.
\newblock Classics in Mathematics. Springer.

\end{thebibliography}
\vspace{-.1cm}
\enlargethispage{5cm}
\end{document}